%% file: main.tex
\documentclass[11pt]{article}
\usepackage{fullpage}
\usepackage{mathtools,amsfonts,amsthm,amssymb,bbm}
\usepackage{xcolor}
\usepackage[backref,colorlinks,citecolor=blue,linkcolor=magenta,bookmarks=true]{hyperref}
\usepackage{cleveref}
\usepackage[algo2e,ruled,linesnumbered]{algorithm2e}
\usepackage{authblk}
\usepackage{float}
\usepackage{graphicx}
\usepackage{subcaption}

\usepackage{color-edits}
\addauthor{mq}{blue}

\newcommand{\1}[1]{\mathbbm{1}\left[#1\right]} 

\newcommand{\bbR}{\mathbb{R}} 
\newcommand{\Bern}{\mathsf{Bernoulli}} 
\newcommand{\Binomial}{\mathsf{Binomial}} 
\newcommand{\Cal}{\mathcal{C}}
\newcommand{\calA}{\mathcal{A}} 
\newcommand{\calD}{\mathcal{D}} 
\newcommand{\CalDist}{\mathsf{CalDist}}
\newcommand{\calI}{\mathcal{I}} 
\newcommand{\calN}{\mathcal{N}} 
\newcommand{\calP}{\mathcal{P}} 
\newcommand{\calX}{\mathcal{X}}

\newcommand{\coNP}{\textsf{co-NP}}
\newcommand{\cost}{\mathsf{cost}}
\newcommand{\lcost}{\cost^{\mathsf{L}}}
\newcommand{\rcost}{\cost^{\mathsf{R}}}
\newcommand{\Dhat}{\widehat{\calD}}
\newcommand{\Dmix}{\calD_{\mathsf{mixed}}}
\newcommand{\Dpure}{\calD_{\mathsf{pure}}}
\newcommand{\dTV}{d_{\mathsf{TV}}}
\newcommand{\eps}{\epsilon}
\newcommand{\Ex}[2]{\operatorname*{\mathbb{E}}_{#1}\left[#2\right]} 
\newcommand{\istar}{i^{\star}}
\newcommand{\NP}{\mathsf{NP}}

\newcommand{\opt}{\mathsf{opt}}
\newcommand{\poly}{\operatorname*{poly}} 

\newcommand{\pr}[2]{\Pr_{#1}\left[#2\right]} 
\newcommand{\pstar}{p^{\star}}
\newcommand{\sgn}{\mathrm{sgn}} 

\newcommand{\uppercost}{\overline{\cost}}
\newcommand{\zo}{\{0, 1\}} 

\newtheorem{theorem}{Theorem}
\newtheorem{definition}{Definition}
\newtheorem{corollary}{Corollary}
\newtheorem{lemma}{Lemma}
\newtheorem{fact}{Fact}

\title{Computational and Statistical Hardness of Calibration Distance}
\date{}
\author{Mingda Qiao}
\affil{University of Massachusetts Amherst}
\affil{\texttt{mqiao@umass.edu}}

\begin{document}
    \begin{titlepage}
        \maketitle
        \thispagestyle{empty}
        \begin{abstract}
            The distance from calibration, introduced by B\l{}asiok, Gopalan, Hu, and Nakkiran~(STOC 2023), has recently emerged as a central measure of miscalibration for probabilistic predictors. We study the fundamental problems of \emph{computing} and \emph{estimating} this quantity, given either an exact description of the data distribution or only sample access to it.
    
            We give an efficient algorithm that exactly computes the calibration distance when the distribution has a uniform marginal and noiseless labels, which improves the $O(1/\sqrt{|\calX|})$ additive approximation of Qiao and Zheng~(COLT 2024) for this special case. Perhaps surprisingly, the problem becomes $\NP$-hard when \emph{either} of the two assumptions is removed. We extend our algorithm to a polynomial-time approximation scheme for the general case.
            
            For the estimation problem, we show that $\Theta(1/\eps^3)$ samples are sufficient and necessary for the empirical calibration distance to be upper bounded by the true distance plus $\eps$. In contrast, a polynomial dependence on the domain size---incurred by the learning-based baseline---is unavoidable for two-sided estimation.
    
            Our positive results are based on simple sparsifications of both the distribution and the target predictor, which significantly reduce the search space for computation and lead to stronger concentration for the estimation problem. To prove the hardness results, we introduce new techniques for certifying lower bounds on the calibration distance---a problem that is hard in general due to its $\coNP$-completeness.
        \end{abstract}
    \end{titlepage}

    \input{introduction}

    \input{prelim}

    \input{overview}

    \input{discussion}

    \input{algorithm}

    \input{hardness}

    \input{sample_UB}

    \input{sample_LB}

    \newpage
    \appendix

    \input{appendix}

	\bibliographystyle{alpha}
	\bibliography{main}
\end{document}

%% file: introduction.tex
\section{Introduction}
Calibration is a well-known criterion for reliable probabilistic forecasts~\cite{Brier50,Dawid82} and, more recently, trustworthy machine learning models~\cite{GPSW17}. A predictor is \emph{perfectly calibrated} if, conditioning on predicting each value $\alpha \in [0, 1]$, the conditional probability of a positive outcome is exactly $\alpha$. Predictors that are far from calibrated---e.g., weather forecasters that predicted a $90\%$ chance of rain on $100$ days among which only $50$ were rainy---are intuitively biased and not trustworthy. Calibrated predictions are known to have many nice properties, such as no-regret downstream decision-making and convergence to correlated equilibria in game dynamics~\cite{FV97}.

While the notion of perfect calibration is widely accepted, there is no consensus on the ``right'' \emph{calibration measure} for quantifying miscalibration. B\l{}asiok, Gopalan, Hu, and Nakkiran~\cite{BGHN23} recently initiated the systematic study of calibration measures and proposed a notion of \emph{distance from calibration} as a ``gold standard''. The distance from calibration is defined as the distance to the closest perfectly calibrated predictor in the $\ell_1$ metric induced by the marginal distribution. The authors of~\cite{BGHN23} further examined existing calibration measures and identified several \emph{consistent} measures that are always polynomially related to the calibration distance.

Perhaps surprisingly, fundamental problems around this calibration metric remain largely open:
\begin{quote}
    \centering \emph{Can the calibration distance be efficiently computed, either exactly or approximately?\\What is the sample complexity of estimating the calibration distance?}
\end{quote}
The former computational problem assumes complete knowledge of the data distribution, while only samples from the distribution are available in the latter problem. Indeed, positive answers to these questions would make the calibration distance \emph{per se} both a useful error metric and a natural objective for training prediction models. However, it was noted in~\cite{BGHN23} that ``it is not clear how to compute this distance efficiently'', since the set of perfectly calibrated predictors ``is non-convex, and in fact it is discrete when the domain $\calX$ is discrete''.

The situation is in stark contrast to almost all the other existing calibration measures, which can be efficiently computed by definition. For instance, the \emph{smooth calibration error} introduced by Kakade and Foster~\cite{KF08}---one of the consistent calibration measures---can be computed by a natural linear program and estimated up to error $\eps$ using $O(1/\eps^2)$ samples~\cite{BGHN23}. Hu, Jambulapati, Tian, and Yang~\cite{HJTY24} showed that the \emph{lower distance from calibration}---a continuous relaxation of the calibration distance introduced by~\cite{BGHN23}---can be efficiently tested, both computationally and statistically.

\subsection{Problem Setup}
Formally, we consider binary classification with a finite domain $\calX$ and distribution $\calD$ over $\calX \times \zo$. A predictor $f: \calX \to [0, 1]$ is \emph{perfectly calibrated} with respect to $\calD$ if it holds for every $\alpha \in [0, 1]$ (with $\pr{(x, y) \sim \calD}{f(x) = \alpha} > 0$) that $\pr{(x, y) \sim \calD}{y = 1 \mid f(x) = \alpha} = \alpha$. We write $\Cal(\calD)$ as the family of all perfectly calibrated predictors with respect to $\calD$.

Let $\calD_x$ be the $\calX$-marginal of $\calD$. For each $x \in \calD$, $\mu_{\calD}(x) \coloneqq \pr{(x', y) \sim \calD}{y = 1 \mid x' = x}$ denotes the conditional probability of label $1$ given $x$.\footnote{We may define $\mu_{\calD}(x) = 0$ in the degenerate case that $\calD_x(x) = 0$, as such values have no effect.} For $\calX' \subseteq \calX$, we also write
\[
    \mu_{\calD}(\calX') \coloneqq \pr{(x, y) \sim \calD}{y = 1 \mid x \in \calX'} = \frac{\sum_{x \in \calX'}\calD_x(x) \cdot \mu_{\calD}(x)}{\sum_{x \in \calX'}\calD_x(x)}
\]
as the conditional probability given $x \in \calX'$. The subscript in $\mu_{\calD}$ will be omitted when the distribution is clear from the context. We consider the following two families of distributions: $\calD$ is \emph{uniform} if $\calD_x$ is the uniform distribution over $\calX$, and $\calD$ is \emph{noiseless} if $\mu_{\calD}(x) \in \{0, 1\}$ holds for every $x \in \calX$.

\paragraph{Distance from calibration.} Let $d_{\calD}(f, g) \coloneqq \Ex{x \sim \calD_x}{|f(x) - g(x)|}$ denote the $\ell_1$ distance between predictors $f, g: \calX \to [0, 1]$ over $\calD$. The distance from calibration (or calibration distance) of $f$ over $\calD$ is defined as $\CalDist_{\calD}(f) \coloneqq \inf_{g \in \Cal(\calD)}d_{\calD}(f, g)$.

In the computational problem, the algorithm is given the probability mass function of $\calD$ as well as the predictor $f$ as tables of values. In the estimation version, the algorithm can sample from $\calD$ and query predictor $f$ on domain elements of its choice. The goal in both problems is to output $\CalDist_{\calD}(f)$, either exactly or up to an additive error of $\eps > 0$.

\paragraph{Model of computation.} For the hardness proofs, we assume that the probability masses of $\calD$ and the predictions made by $f$ are all rational numbers. The size of the instance is the total number of bits needed for representing the $\Theta(|\calX|)$ values in reduced form.

For brevity, the algorithms are stated in the real RAM model, where inputs are given as exact real numbers, and each basic arithmetic operation or comparison takes constant time. In all of our algorithms, every intermediate value can be computed from the inputs in $O(|\calX|)$ basic operations. Therefore, if all inputs are rational, the algorithms can be exactly implemented in alternative models (such as the word RAM) with at most an $\poly(|\calX|)$ blow-up in the runtime.

\subsection{Overview of Results}
We give a near-complete characterization of the computational and statistical landscapes for the distance from calibration. Our main algorithmic result is an efficient algorithm when $\calD$ is both uniform and noiseless, and a polynomial-time approximation scheme (PTAS) for the general case.
\begin{theorem}\label{thm:runtime}
    Given $\calD$ and $f$, $\CalDist_{\calD}(f)$ can be computed exactly in $O(|\calX|^4)$ time when $\calD$ is both uniform and noiseless. Moreover, it can be approximated up to an error of $\eps$ in $|\calX|^{O(1/\eps)}$ time assuming that $\calD$ is uniform, and in $(|\calX| / \eps)^{O(1/\eps)}$ time in general.
\end{theorem}

When $\calD$ is both uniform and noiseless, the problem is equivalent to computing the distance from calibration in the \emph{sequential} setting over $T = |\calX|$ time steps. The best known result for this special case is a $\poly(T)$-time approximation algorithm due to Qiao and Zheng~\cite{QZ24} that guarantees an $O(1 / \sqrt{T})$ additive error.

We complement \Cref{thm:runtime} with the following hardness result, which shows that the assumption on $\calD$ is essentially minimal for efficient exact computation.

\begin{theorem}\label{thm:NP-hard}
    It is $\NP$-hard to compute the calibration distance exactly, both when $\calD$ is noiseless and when $\calD$ is uniform.
\end{theorem}

For the estimation problem, we have the following sample upper bounds for a natural estimator: the calibration distance over the empirical distribution $\Dhat$.
\begin{theorem}\label{thm:sample-UB}
    A sample size of $O(|\calX| / \eps^2)$ is sufficient for $|\CalDist_{\Dhat}(f) - \CalDist_{\calD}(f)| \le \eps$ to hold with probability $0.99$. Moreover, a sample size of $O(1 / \eps^3)$ is sufficient for $\CalDist_{\Dhat}(f) \le \CalDist_{\calD}(f) + \eps$ to hold with probability $0.99$.
\end{theorem}

The first part of \Cref{thm:sample-UB} gives a baseline that matches the sample complexity of learning $\calD$ up to $\eps$ error in the total variation (TV) distance. The second part suggests that, for moderate values of $\eps$, a much smaller sample of size independent of $|\calX|$ suffices for avoiding over-estimation.

One may hope to prove an analogous bound for under-estimation, thereby improving the sample complexity of estimation to $\poly(1/\eps)$. Such a result would have an immediate application to the computational problem: By drawing a size-$\poly(1/\eps)$ sample and running the PTAS on the empirical distribution, we can further improve the runtime to $O(|\calX|) + 2^{\widetilde O(1/\eps)}$.

Unfortunately, our next lower bound shows that this is impossible---a polynomial dependence on the domain size is unavoidable. Moreover, the $1/\eps^3$ dependence in the one-sided convergence result is tight up to constant factors.

\begin{theorem}[Informal versions of \Cref{thm:sample-LB-formal,thm:one-sided-LB-formal}]\label{thm:sample-LB}
    The sample complexity of estimating the calibration distance up to error $\eps$ is at least $\Omega(\sqrt{|\calX|} / \eps)$. Moreover, when $|\calX| = \Omega(1/\eps)$, at least $\Omega(1/\eps^3)$ samples are necessary for $\CalDist_{\Dhat}(f) \le \CalDist_{\calD}(f) + \eps$ to hold.
\end{theorem}

The assumption $|\calX| = \Omega(1/\eps)$ for the second part is necessary; otherwise, the first part of \Cref{thm:sample-UB} gives an upper bound of $O(|\calX| / \eps^2) \ll 1/\eps^3$.

\subsection{Related Work}
\paragraph{Distance from calibration.} The distance from calibration was introduced by B\l{}asiok, Gopalan, Hu, and Nakkiran~\cite{BGHN23} in their systematic study of calibration measures. They identified several ``consistent'' calibration measures that are polynomially related to the distance from calibration on every distance, while satisfying other natural properties such as continuity. The work of~\cite{BGHN23} primarily focused on the \emph{prediction-only access} (PA) model, in which only prediction-label pairs of form $(f(x), y)$ (where $(x, y) \sim \calD$) are observed. It was shown that the distance from calibration is not identifiable in the PA model. In comparison, our lower bounds (\Cref{thm:NP-hard,thm:sample-LB}) establish further barriers---both computational and statistical---to computing the calibration distance in stronger access models, where one is given either $\calD$ or triples of form $(x, f(x), y)$.

Independent and concurrent work of Gopalan, Stavropoulos, Talwar, and Tankala~\cite{GSTT26} shows that, in the PA model, $\Omega(\sqrt{|\calX|})$ samples are necessary for distinguishing whether the \emph{upper distance from calibration}---another consistent calibration measure introduced by~\cite{BGHN23}---is $\Omega(\eps)$ or $O(\eps^2)$. Their lower bound follows from a similar collision argument to the one for the first part of \Cref{thm:sample-LB}.

\paragraph{Sequential calibration.} The calibration distance has also been studied in the sequential prediction setting, where the forecaster makes probabilistic predictions on $T$ binary outcomes that are observed sequentially. Qiao and Zheng~\cite{QZ24} gave an efficient algorithm that approximates the calibration distance up to $O(1/\sqrt{T})$ error, and showed that the optimal error rate is between $O(T^{-1/2})$ and $\Omega(T^{-2/3})$ when the sequence is chosen by an adaptive adversary. Arunachaleswara, Collina, Roth, and Shi~\cite{ACRS25} obtained an efficient and deterministic algorithm that achieves the $O(T^{-1/2})$ upper bound with a leading constant factor of $2$.

\paragraph{Other calibration measures.} Subsequent work introduced calibration measures that satisfy other natural properties such as truthfulness~\cite{HQYZ24,QZ25,HWW25}, testability~\cite{RSBRW25}, and guarantees for downstream decision-making~\cite{KPLST23,NRRX23,RS24,HW24,QZ25,RSBRW25,HWY25}. The optimal error rates in the sequential calibration setting were also studied for several calibration measures, including the ($\ell_1$) expected calibration error (ECE)~\cite{FV98,QV21,DDFGKO25} and its high-dimensional variant~\cite{Peng25,FGMS25}, $\ell_2$-calibration~\cite{FKOPLST25}, U-calibration~\cite{KPLST23,LSS24}, and KL-calibration~\cite{LSS25}.

\paragraph{Hardness of calibration.} On the technical level, closest to our study is a work of Okati, Tsirtsis, and Gomez Rodriguez~\cite{OTGR23} on the problem of partitioning the domain into subgroups subject to a within-group monotonicity constraint given the predictor. They proved that the problem is $\NP$-hard by reduction from the partition problem, and gave an efficient dynamic programming algorithm under an additional constraint that each subgroup must form a contiguous interval.

In comparison, our hardness proofs are more complicated due to the optimization (instead of decision) nature of the problem. Indeed, the crux of the analysis is to argue that, in every ``No'' instance of Partition (or variants thereof), the non-existence of ``good'' partitions leads to a strictly higher distance from calibration. On the algorithmic side, part of our results is based on a similar restriction to contiguous partitions, though the restriction is without loss of generality after proper sparsification of the instance.

\subsection{Organization of the Paper}
A few useful properties of the calibration distance are given in \Cref{sec:prelim}. In \Cref{sec:overview}, we present the proof sketches for our results while highlighting our techniques, followed by a discussion of open problems in \Cref{sec:discussion}. The full proofs of the computational results (\Cref{thm:runtime,thm:NP-hard}) are presented in \Cref{sec:algorithm,sec:hardness}, while the estimation results (\Cref{thm:sample-UB,thm:sample-LB}) are proved in \Cref{sec:sample-UB,sec:sample-LB,sec:one-sided-LB}.

%% file: prelim.tex
\section{Preliminaries}\label{sec:prelim}

\subsection{Calibrated Predictors as Partitions}
Perfectly calibrated predictors naturally correspond to partitions of the domain.
\begin{fact}\label{fact:CalDist-vs-partition}
    Every partition $\{\calX_1, \calX_2, \ldots, \calX_k\}$ of $\calX$ induces a perfectly calibrated predictor
    \[
        f(x) = \mu_{\calD}(\calX_i), ~\forall i \in [k], x \in \calX_i.
    \]
    Conversely, for every $f \in \Cal(\calD)$, the partition
    \[
        \{f^{-1}(\alpha): \alpha \in [0, 1]\}
    =   \{\{x \in \calX: f(x) = \alpha\}: \alpha \in [0, 1]\}
    \]
    induces the predictor $f$ as above.
\end{fact}
\Cref{fact:CalDist-vs-partition} implies a surjection from the families of partitions to $\Cal(\calD)$. Since $\calX$ is finite, the set $\Cal(\calD)$ is also finite, so the infimum in the definition of $\CalDist_{\calD}(f)$ can always be achieved. Moreover, computing the distance from calibration can be viewed as an optimization problem over partitions of $\calX$. The following notion of the \emph{cost} of a partition will be useful.

\begin{definition}[Cost of partition]\label{def:partition-cost}
    Given $\calD$ and $f$, the cost of $\calX' \subseteq \calX$ is
    \[
        \cost_{\calD, f}(\calX') \coloneqq \sum_{x \in \calX'}\calD_x(x) \cdot \left|f(x) - \mu_{\calD}(\calX')\right|.
    \]
    For a partition $\calP = \{\calX_1, \calX_2, \ldots, \calX_k\}$ of $\calX$, we abuse the notation and write
    \[
        \cost_{\calD, f}(\calP) \coloneqq \sum_{i=1}^{k}\cost_{\calD, f}(\calX_i).
    \]
\end{definition}
Note that $\cost_{\calD, f}(\calP)$ is exactly the distance $d_{\calD}(f, g)$, where $g \in \Cal(\calD)$ is the predictor induced by $\calP$. Therefore, $\CalDist_{\calD}(f)$ is the minimum of $\cost_{\calD, f}(\calP)$ over all partitions of $\calX$.

\subsection{Continuity of the Calibration Distance}

By definition, $\CalDist_{\calD}(f)$ is $1$-Lipschitz in $f$ with respect to $d_{\calD}$.

\begin{fact}\label{fact:continuity-in-f}
    For any distribution $\calD \in \Delta(\calX \times \zo)$ and predictors $f, g: \calX \to [0, 1]$,
    \[\left|\CalDist_{\calD}(f) - \CalDist_{\calD}(g)\right| \le d_{\calD}(f, g).\]
\end{fact}

We prove similar Lipschitz properties in the distribution $\calD$ with respect to the TV distance. In Appendix~\ref{app:basic}, we prove that the cost of every partition is Lipschitz in $\calD$.

\begin{lemma}\label{lem:partition-cost-continuity-in-D}
    For any $\calD_1, \calD_2 \in \Delta(\calX \times \zo)$, $f: \calX \to [0, 1]$, and partition $\calP$ of $\calX$,
    \[
        \left|\cost_{\calD_1, f}(\calP) - \cost_{\calD_2, f}(\calP)\right| \le 5 \cdot \dTV(\calD_1, \calD_2).
    \]
\end{lemma}

It immediately follows that the distance from calibration is also Lipschitz in $\calD$.
\begin{lemma}\label{lem:continuity-in-D}
    For any $\calD_1, \calD_2 \in \Delta(\calX \times \zo)$ and $f: \calX \to [0, 1]$,
    \[
        \left|\CalDist_{\calD_1}(f) - \CalDist_{\calD_2}(f)\right| \le 5 \cdot \dTV(\calD_1, \calD_2).
    \]
\end{lemma}

\begin{proof}
    By \Cref{fact:CalDist-vs-partition}, $\CalDist_{\calD_i}(f) = \min_{\calP}\cost_{\calD_i, f}(\calP)$ holds for each $i \in \{1, 2\}$, where $\calP$ ranges over all partitions of $\calX$. The lemma then follows from \Cref{lem:partition-cost-continuity-in-D}.
\end{proof}

%% file: overview.tex
\section{Technical Overview}\label{sec:overview}
\subsection{Calibration Distance as Optimal Transport}
It will be helpful to view the distance from calibration from an optimal transport perspective. Suppose that there are $|\calX|$ point masses on the line segment $[0, 1]$: each element $x \in \calX$ corresponds to a point of \emph{mass} $\calD_x(x)$ and \emph{value} $\mu_{\calD}(x)$ located at $f(x)$. In the definition of $\CalDist_{\calD}(f)$, the alternative predictor $g$ specifies the destinations of these elements: each point $x$ is transported from location $f(x)$ to $g(x)$ at a cost of $\calD_x(x) \cdot |f(x) - g(x)|$---the product of the mass and the distance. The constraint $g \in \Cal(\calD)$ corresponds to a requirement for the resulting configuration: at any location $\alpha \in [0, 1]$, the mass-weighted average value---namely, the average of $\mu_{\calD}(x)$ with weights proportional to $\calD_x(x)$---should be equal to $\alpha$.

Given the equivalence above, in the rest of this section, we will sometimes specify the perfectly calibrated predictor $g$ by describing a way of ``moving'' the point masses defined by the problem instance $(\calD, f)$.

\subsection{Sparsification for Efficient Computation}\label{sec:overview-algorithm}
As a warm-up, we prove the first part of \Cref{thm:runtime}. From the optimal transport perspective, each of the $|\calX|$ points has mass $\calD_x(x) = 1/|\calX|$ and binary value $\mu_{\calD}(x) \in \{0, 1\}$. Let $x^{(0)}_1, x^{(0)}_2, \ldots, x^{(0)}_{n_0} \in \calX$ be the $0$-valued points from left to right, i.e., $f(x^{(0)}_1) \le f(x^{(0)}_2) \le \cdots \le f(x^{(0)}_{n_0})$. Define $n_1 = |\calX| - n_0$ and $x^{(1)}_1, x^{(1)}_2, \ldots, x^{(1)}_{n_1}$ for $1$-valued points analogously.

Let $g \in \Cal(\calD)$ be a predictor that witnesses $\CalDist_{\calD}(f)$. Let $p^{(0)}_1 \le p^{(0)}_2 \le \cdots \le p^{(0)}_{n_0}$ (resp., $p^{(1)}_1 \le p^{(1)}_2 \le \cdots \le p^{(1)}_{n_1}$) denote the destinations of the $0$-valued (resp., $1$-valued) points specified by $g$. Since the optimal transport on the line is order-preserving, we may assume that each point mass $x^{(b)}_i$ is moved to $p^{(b)}_i$.

Switching to the partition view of the calibration distance (\Cref{fact:CalDist-vs-partition}), the observation above allows us to focus on \emph{contiguous} partitions of $\calX$. The part that corresponds to the largest predicted value must be of form $\{x_{m_0 + 1}^{(0)}, x_{m_0 + 2}^{(0)}, \ldots, x_{n_0}^{(0)}\} \cup \{x_{m_1 + 1}^{(1)}, x_{m_1 + 2}^{(1)}, \ldots, x_{n_1}^{(1)}\}$. Then, the part with the second largest prediction must be the union of a suffix of $x^{(0)}_1, x^{(1)}_2, \ldots, x^{(0)}_{m_0}$ and a suffix of $x^{(1)}_1, x^{(1)}_2, \ldots, x^{(1)}_{m_1}$, and so on.

Therefore, the optimal partition can be found in $O(|\calX|^5)$ time via straightforward dynamic programming. Let $\opt(m_0, m_1)$ denote the optimal cost among all contiguous partitions of the $m_0$ left-most $0$-valued elements and the $m_1$ left-most $1$-valued ones. There are $O(|\calX|^2)$ such values to be computed, while each $\opt(m_0, m_1)$ can be found by enumerating $O(m_0 m_1) = O(|\calX|^2)$ choices of the right-most part. Furthermore, the cost of the right-most part can be computed in $O(|\calX|)$ time. With careful pre-processing, we can compute each cost in constant time, thereby improving the runtime to $O(|\calX|^4)$.

\paragraph{A QPTAS via type-sparsity.} The key observation above was that all $0$-valued (resp., $1$-valued) elements are indistinguishable: swapping any two $0$-valued (resp., $1$-valued) points in the resulting configuration maintains calibration. This observation allows us to focus on contiguous partitions, thereby reducing the number of sub-problems to $O(|\calX|^2)$.

This approach naturally extends to distributions with only a few \emph{types} of indistinguishable elements. Formally, $x_1, x_2 \in \calX$ are of the same type if $\calD_x(x_1) = \calD_x(x_2)$ and $\mu_{\calD}(x_1) = \mu_{\calD}(x_2)$. We may restrict our attention to \emph{type-monotone} predictors that are order-preserving within each type. Assuming there are at most $k$ different types, we can compute $\CalDist_{\calD}(f)$ using a dynamic programming algorithm with $O(|\calX|^k)$ states in $O(|\calX|^{2k})$ time.

We handle the general case by approximating the distribution $\calD$ with a \emph{type-sparse} one that is $O(\eps)$-close in the TV distance. By \Cref{lem:continuity-in-D}, this gives an additive approximation with error $\eps$. If $\calD$ has a uniform $\calX$-marginal, it suffices to round the conditional probabilities to integer multiples of $O(\eps)$, resulting in an $O(1/\eps)$-type-sparse distribution. For general $\calD$, we additionally discretize the marginal probabilities by removing all probability masses below $O(\eps / |\calX|)$ and rounding the remaining ones to integer powers of $1 + O(\eps)$. This leads to at most $O(\log(|\calX| / \eps) / \eps)$ different marginal probabilities, and thus a type-sparsity of $O(\log(|\calX| / \eps) / \eps^2)$.

The runtime of the resulting approximation algorithm is $|\calX|^{O(1/\eps)}$ in the uniform case, which gives the second part of \Cref{thm:runtime}. In the general case, we obtain a runtime of $|\calX|^{O(\log(|\calX| / \eps) / \eps^2)}$ and thus only a quasi-polynomial-time approximation scheme (QPTAS).

\paragraph{A PTAS via prediction-sparsity.} We obtain a PTAS for the general case by exploiting a different type of sparsity. Let $g \in \Cal(\calD)$ be a predictor that witnesses the calibration distance. While $g$ might predict many different values, we can reduce the size of its range while maintaining both perfect calibration and the distance from $f$.

We partition the interval $[0, 1]$ into $O(1/\eps)$ intervals of length $O(\eps)$. Consider the points in each interval in the configuration specified by $g$, and move all of them to the same location---the mass-weighted average of all point values. This introduces $O(\eps)$ units of additional transportation for each point mass, and thus an $O(\eps)$ increase in the total cost.

Therefore, we may restrict our attention to partitions of the following form: There are exactly $O(1/\eps)$ parts, one for each length-$O(\eps)$ interval. Furthermore, each interval corresponds to a part with a mass-weighted average value that falls into that interval. Intuitively, optimizing over such partitions resembles a multiple knapsack problem with $O(1/\eps)$ knapsacks, and should admit a PTAS. There are two remaining difficulties: First, to verify the validity of each part, we need to track the mass-weighted average value. Second, the transportation cost for a point depends on its final destination (namely, the mass-weighted average), which is not determined until we identify all elements in the part.

The first can be easily handled by rounding the probability masses of $\calD$ to integer multiples of $O(\eps / |\calX|)$ at the cost of an $O(\eps)$ change to the distribution in the TV distance, and thus the calibration distance (by \Cref{lem:continuity-in-D}). We address the second issue by using the middle point of the interval as a proxy for the destination. As long as the actual destination is inside the interval, doing so only changes the cost by $O(\eps)$, which is again absorbed into the approximation error.

Therefore, each state in the dynamic programming needs to track the total $0$-valued and $1$-valued masses in each of the $O(1/\eps)$ parts. Due to the discretization, there are at most $(|\calX| / \eps)^{O(1/\eps)}$ states. The transitions are naturally determined by enumerating the part to which the element in question belongs. This leads to a PTAS with runtime $(|\calX| / \eps)^{O(1/\eps)}$ and completes \Cref{thm:runtime}.

\subsection{$\NP$-Hardness from Subset Sum}\label{sec:overview-hardness}
In light of the connection between the calibration distance and set partitions (\Cref{fact:CalDist-vs-partition}), we show that the problem is intractable by reduction from partition-type problems like the Subset Sum Problem (SSP). The numbers in an SSP instance can be naturally embedded into the various probabilities in a calibration distance instance.

We will start by reviewing an example introduced by~\cite{BGHN23} for separating the calibration distance from other calibration measures. We then show how similar constructions can be used to embed SSP instances, while sketching several techniques for certifying a large calibration distance.

\paragraph{Example of~\cite{BGHN23}.} Suppose that $\calD$ has a uniform marginal over $\calX = \{x_{0^-}, x_{1^-}, x_{0^+}, x_{1^+}\}$ and conditional probabilities $\mu_{\calD}(x_{0^-}) = \mu_{\calD}(x_{0^+}) = 0$ and $\mu_{\calD}(x_{1^-}) = \mu_{\calD}(x_{1^+}) = 1$. The predictions are given by $f(x_{0^-}) = f(x_{1^-}) = 1/2 - \eps$ and $f(x_{0^+}) = f(x_{1^+}) = 1/2 + \eps$. From the transportation perspective, we have four points of equal mass $1/4$. The bit ($0$ or $1$) indicates the value of the point, while the sign ($-$ or $+$) denotes whether the initial location is slightly below or above $1/2$. When $\eps$ is small enough, the optimal transport is given by moving all points to $1/2$ at a total cost of $\eps$, since any other strategy must move at least one point to $\{0, 1/3, 2/3, 1\}$ and incur a cost $\gg \eps$. See \Cref{fig:CalDist} for a visualization of this instance.

\begin{figure}[!htbp]
    \centering

    \begin{subfigure}[b]{0.45\textwidth}
        \centering
        \includegraphics[width=\textwidth]{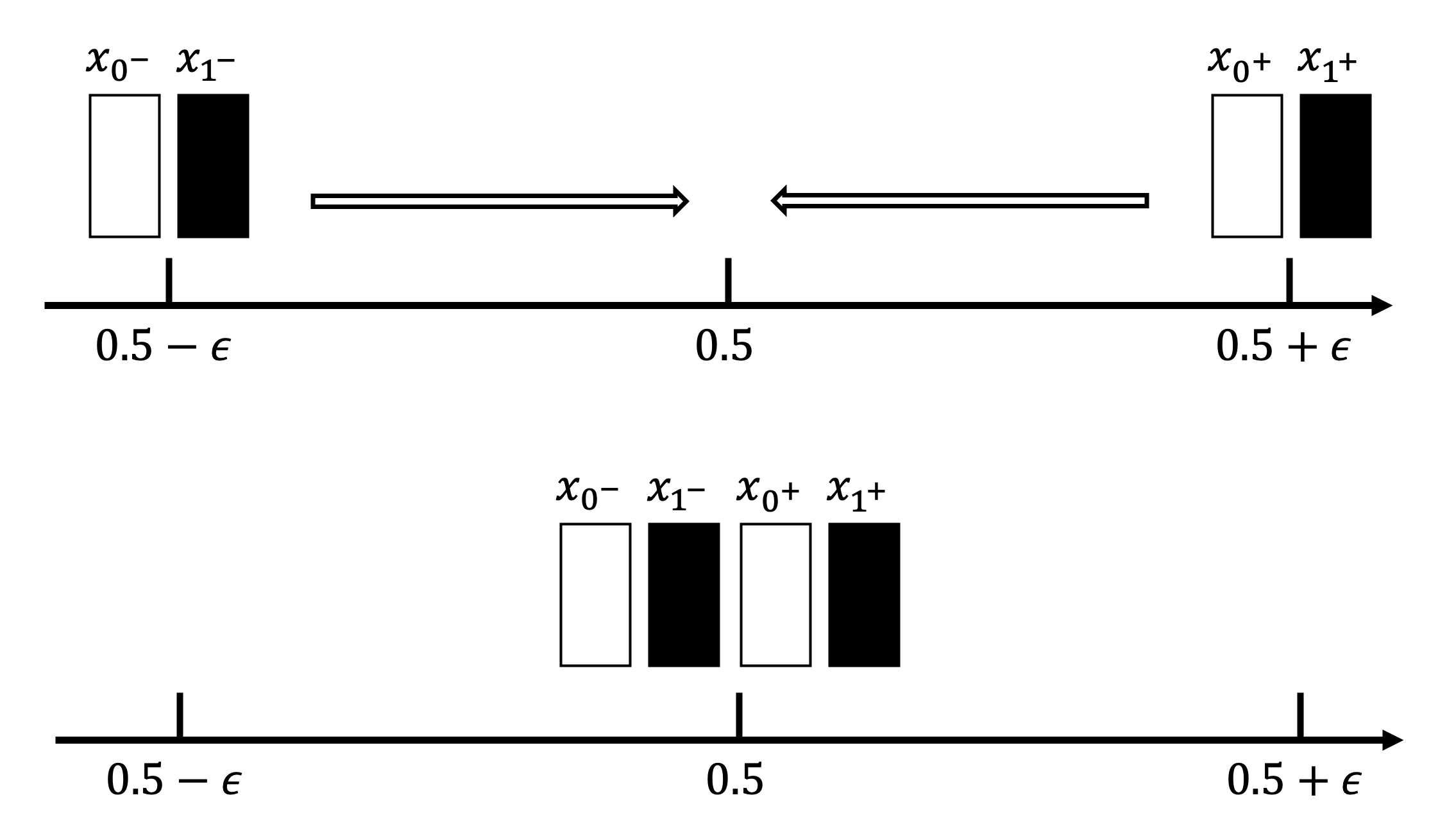}
        \caption{Top: Starting configuration of the instance $(\calD, f)$. Bottom: Resulting configuration given by the predictor that witnesses $\CalDist_{\calD}(f)$.}
        \label{fig:CalDist}
    \end{subfigure}
    \hfill
    \begin{subfigure}[b]{0.45\textwidth}
        \centering
        \includegraphics[width=\textwidth]{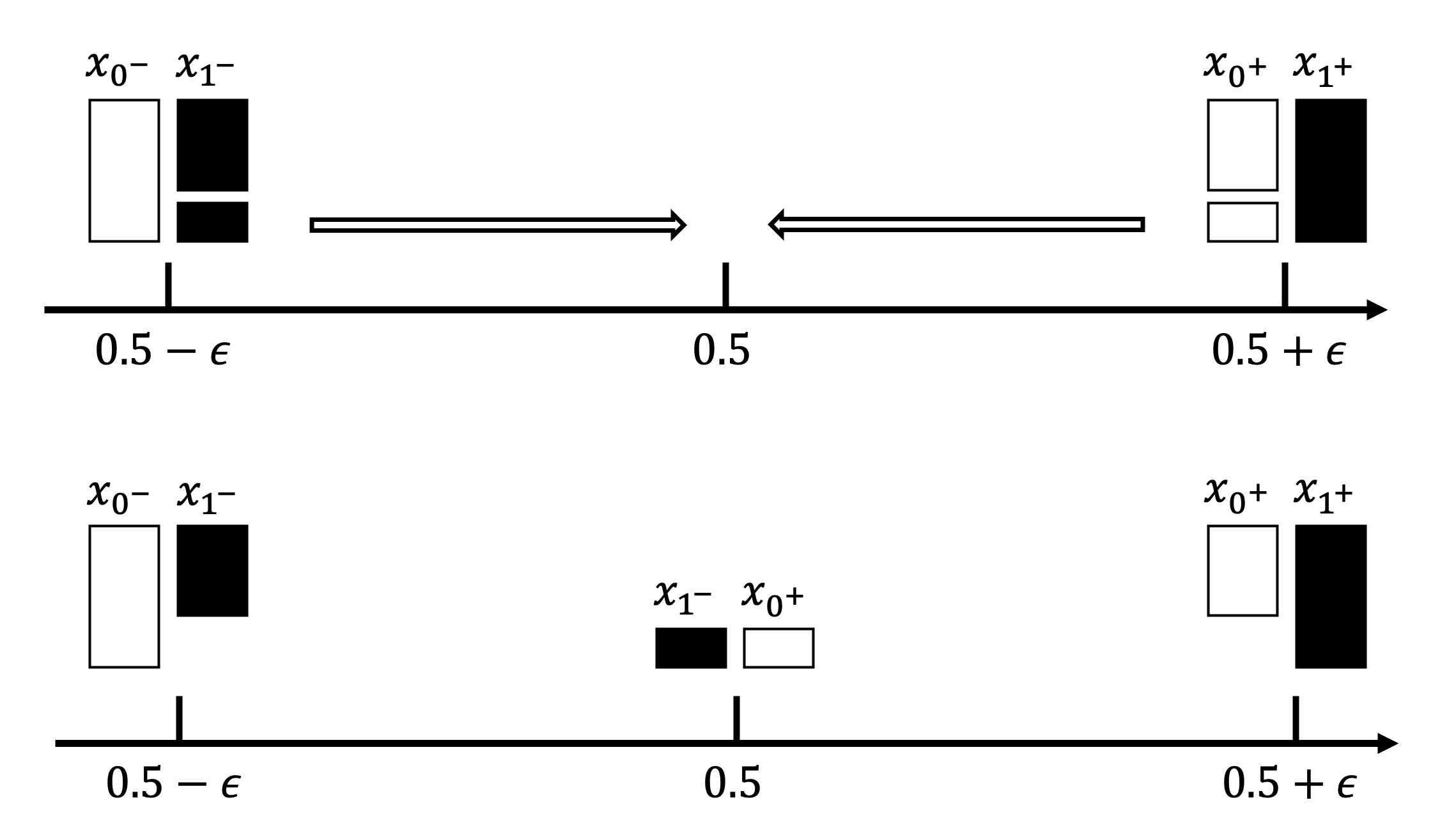}
        \caption{The optimal transport when elements are decomposable. Each of $x_{1^-}$ and $x_{0^+}$ is broken into two pieces; the smaller piece has mass $\delta = \Theta(\eps)$.}
        \label{fig:LowerCalDist}
    \end{subfigure}

    \medskip

    \begin{subfigure}[b]{0.75\textwidth}
        \centering
        \includegraphics[width=\textwidth]{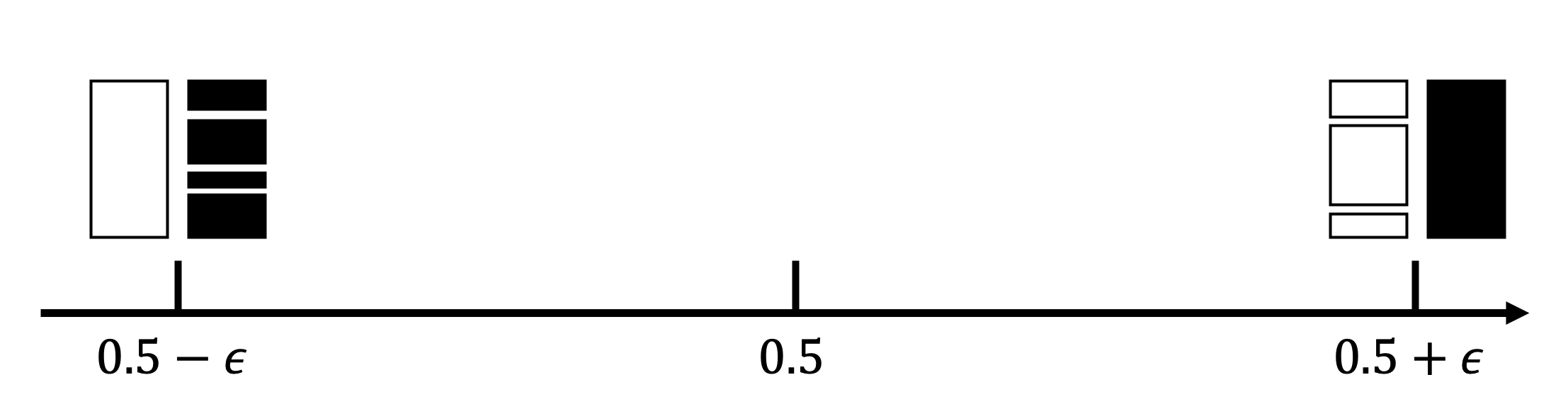}
        \caption{A potential reduction from the Subset Sum Problem via replacing $x_{1^-}$ and $x_{0^+}$ with multiple elements. The actual proof uses a slightly different construction (see \Cref{def:reduction-noiseless}) to simplify the analysis.}
        \label{fig:Reduction}
    \end{subfigure}

    \caption{Visualization of the calibration distance (\Cref{fig:CalDist}) and the lower distance from calibration (\Cref{fig:LowerCalDist}) on the instance from \cite[Lemma 4.5]{BGHN23}, as well as a potential reduction from Subset Sum (\Cref{fig:Reduction}). Rectangles denote elements of $\calX$, with sizes corresponding to marginal probabilities and colors (black or white) corresponding to conditional probabilities ($1$ or $0$). The locations of rectangles denote the ($[0, 1]$-valued) predictions on the elements, issued by either the original predictor $f$ or the alternative predictor witnessing the calibration distance.}
\end{figure}

Imagine, for now, that each point mass is broken into infinitely many points of infinitesimal masses, each with the same value and at the same location. Equivalently, we can freely break each point mass into pieces, and transport each piece to a different location.\footnote{This corresponds to the \emph{lower distance from calibration} introduced by~\cite{BGHN23}.} Then, for some $\delta = \Theta(\eps)$, once we remove a $\delta$-mass piece of $x_{1^-}$ away from its location $1/2 - \eps$, the calibration condition is satisfied at $1/2 - \eps$. Similarly, we may remove a $\delta$-portion of $x_{0^+}$ to satisfy the condition at $1/2 + \eps$. The two removed pieces have identical masses and opposite values, so we can achieve calibration by moving both pieces to $1/2$; see \Cref{fig:LowerCalDist}. The total cost is then $2\delta\eps = \Theta(\eps^2)$, which is much lower than the original calibration distance of $\eps$.

Suppose, instead, that $x_{1^-}$ and $x_{0^+}$ are broken into finitely many point masses (\Cref{fig:Reduction}). Then, whether the distance from calibration is closer to $\eps$ or $2\delta\eps$ depends on the ``granularity'' of the point masses. If there is a subset with a total mass of exactly $\delta$ on each side, we can still use the strategy above to obtain a calibration distance of $2\delta\eps$. If the opposite is also true, we may reduce the SSP to computing the calibration distance, thereby proving the $\NP$-hardness of the latter problem.

\paragraph{Reduction for the noiseless case.} Recall that an SSP instance consists of positive integers $a_1, a_2, \ldots, a_n$ and $\theta$, and the goal is to decide whether there is a subset with sum $\theta$. In our actual reduction, we slightly modify the example of~\cite{BGHN23}. Let $S$ be the sum of the $n$ numbers. The construction involves small parameters $\pstar, \alpha > 0$ and the following four groups of elements:
\begin{itemize}
    \item Element $x_0$ with $f(x_0) = 1/2$, $\calD_x(x_0) = 1/4 - \pstar/2$, and $\mu(x_0) = 0$.
    \item Element $x_1$ with $f(x_1) = 1/2$, $\calD_x(x_1) = 1/4 + \pstar/2$, and $\mu(x_1) = 1$.
    \item A set of $n$ elements $X_{0^+} \coloneqq \{x_{0^+, 1}, x_{0^+, 2}, \ldots, x_{0^+, n}\}$ with $f(x_{0^+, i}) = 1/2 + \alpha$, $\calD_x(x_{0^+, i}) = \frac{1}{4} \cdot \frac{a_i}{S}$, and $\mu(x_{0^+, i}) = 0$ for every $i \in [n]$.
    \item Element $x_{1^+}$ with $f(x_{1^+}) = 1/2 + \alpha$, $\calD_x(x_{1^+}) = 1/4$, and $\mu(x_{1^+}) = 1$.
\end{itemize}
The parameters $\pstar = \theta / (4S)$ and $\alpha = \pstar / (1 - \pstar)$ are chosen such that, if the numbers in $\{a_i: i \in \calI\}$ sum up to $\theta$, the corresponding set $\{x_{0^+,i}: i \in \calI\} \subseteq X_{0^+}$ has a total mass of $\pstar$. Then, by moving these elements from $1/2 + \alpha$ to $1/2$, we obtain a perfectly calibrated configuration at a cost of $\alpha\pstar$. Note that, by \Cref{fact:CalDist-vs-partition}, this transportation corresponds to the partition
\begin{equation}\label{eq:noiseless-yes-instance-partition}
    \calX = (\{x_0, x_1\} \cup \{x_i: i \in \calI\}) \cup (\{x_{1^+}\} \cup \{x_i: i \in \overline{\calI}\}).
\end{equation}
This suggests a natural reduction from the SSP to the noiseless case: we construct the aforementioned noiseless instance and output ``Yes'' if $\CalDist_{\calD}(f) \le \alpha\pstar$ and ``No'' otherwise.

The harder part is to argue that every ``No'' instance of the SSP leads to $\CalDist_{\calD}(f) > \alpha\pstar$. In principle, we need to lower bound the cost of every partition of $\calX = \{x_0, x_1, x_{1^+}\} \cup X_{0^+}$. Fortunately, we can rule out many inefficient partitions and focus on a few ``structured'' ones. Concretely, every partition with cost $\le \alpha\pstar$ must be of a similar form to \Cref{eq:noiseless-yes-instance-partition}:
\[
    \calX = (\{x_0, x_1\} \cup X_1) \cup (\{x_{1^+}\} \cup X_2) \cup X_3,
\]
where $\{X_1, X_2, X_3\}$ is a partition of $X_{0^+}$.

Then, we write $p_i = \calD_x(X_i)$ for $i \in \{1, 2, 3\}$ and note $p_1 + p_2 + p_3 = \calD_x(X_{0^+}) = 1/4$, so there are only two degrees of freedom in the partition. By expressing the cost as a function of $(p_1, p_2, p_3)$ and analyzing it via elementary but involved calculus, we show that the cost is $\le \alpha\pstar$ only if $p_1 = \pstar$, which is impossible for a ``No'' instance of the SSP.

\paragraph{Reduction for the uniform case.} When $\calD$ has a uniform $\calX$-marginal, a natural idea is to encode SSP instances using the conditional probabilities instead. Both the construction and the analysis of ``No'' instances turn out to be more involved.

We will reduce from a ``balanced'' variant of the SSP, where there are $n = 2k$ positive integers $a_1, a_2, \ldots, a_n$, and the goal is to decide whether there are $k$ numbers that sum up to exactly half of $S \coloneqq \sum_{i=1}^{n}a_i$. In addition, the numbers $a_1, a_2, \ldots, a_n$ are promised to be within a factor of $1 + 1/\poly(k)$. This ``Balanced SSP'' problem is $\NP$-hard by reduction from the partition problem.

Given an instance of Balanced SSP, we set $\eps_i \coloneqq \frac{a_i}{S}$, $\eps \coloneqq \frac{1}{6k}$, and $\delta = \frac{1}{4k}$, and consider an instance $(\calD, f)$ with the following three groups of elements:
\begin{itemize}
    \item $x_1, x_2, \ldots, x_n$ with $f(x_i) = 1/2$ and $\mu(x_i) = 1/2 + \eps_i$.
    \item $x'_1, x'_2, \ldots, x'_n$ with $f(x'_i) = 1/2$ and $\mu(x'_i) = 1/2 - \delta$.
    \item $x''_1, x''_2, \ldots, x''_n$ with $f(x''_i) = 1/2 + \eps$ and $\mu(x''_i) = 1/2$.
\end{itemize}
Again, the easier direction is to show that every ``Yes'' instance of Balanced SSP leads to a small distance from calibration. If $\calI \subseteq [n]$ is a size-$k$ set with $\sum_{i \in \calI}a_i = S / 2$, we can achieve calibration by moving the elements $\{x_i: i \in \overline{\calI}\}$ from $1/2$ to $1/2 + \eps$ at a cost of $\frac{k}{6k} \cdot \eps = \frac{1}{36k}$. This transportation corresponds to the partition
\begin{equation}\label{eq:uniform-yes-instance-partition}
    \calX = \left(\{x_i: i \in \calI\} \cup \{x'_i: i \in [n]\}\right) \cup \left(\{x_i: i \in \overline{\calI}\} \cup \{x''_i: i \in [n]\}\right).
\end{equation}

It remains to show that every ``No'' instance leads to a strictly higher calibration distance. Unlike in the noiseless case, where every efficient partition contains at most three parts, we now need to consider partitions with many parts, as there is no obvious way of reducing the number without increasing the cost. In principle, the cost of each part depends on the numbers of $x$-, $x'$-, and $x''$-elements in it, as well as the sum of $\eps_i$s from the $x$-elements. This leads to a cost that depends on many variables---both discrete and continuous---and is thus hard to analyze.

Our workaround is to use the promise of Balanced SSP, which implies $\eps_i \approx \frac{1}{2k}$ up to a factor of $1 \pm 1/\poly(k)$. Thanks to \Cref{lem:partition-cost-continuity-in-D}, rounding each $\eps_i$ to $\frac{1}{2k}$ only perturbs the distribution (and thus the cost of each partition) by $1/\poly(k)$ additively. After the rounding, there are only three \emph{types} of elements (namely, $x$-, $x'$-, and $x''$-elements) in the sense of \Cref{sec:overview-algorithm}. Each part in the partition can then be specified by its \emph{signature}---the numbers of $x$-, $x'$-, and $x''$-elements.

Recall that the partition in \Cref{eq:uniform-yes-instance-partition} involves two parts with signatures $(k, 2k, 0)$ and $(k, 0, 2k)$ respectively. We consider a slight generalization of such partitions, termed \emph{regular} partitions, in which each part has a signature of form either $(t, 2t, 0)$ or $(t, 0, 2t)$. Our analysis consists of the following three steps:
\begin{itemize}
    \item First, we show that every \emph{irregular} partition has a cost of $\ge \frac{1}{36k} + \Omega(1/k^2)$, via an involved case analysis on all possible signatures. This $\Omega(1/k^2)$ sub-optimality gap dominates the $1/\poly(k)$ error incurred by rounding the $\eps_i$s.
    \item Then, by merging parts in a regular partition, we show that it is without loss of generality to assume that the partition has only two parts with signatures $(k, 2k, 0)$ and $(k, 0, 2k)$ respectively. In other words, every efficient partition must be in the form of \Cref{eq:uniform-yes-instance-partition}.
    \item Finally, the cost becomes a univariate function, e.g., of the sum of $\eps_i$s in the part with signature $(k, 2k, 0)$. It again follows from elementary calculus that a cost of $\le 1/(36k)$ can be achieved only if the Balanced SSP instance admits a solution.
\end{itemize}

\subsection{One-Sided Convergence of the Empirical Distance}
Now, we switch to sample-based estimation and focus on the convergence guarantee of the empirical calibration distance. Let $\Dhat \in \Delta(\calX \times \zo)$ be the empirical distribution of a size-$m$ sample. The first part of \Cref{thm:sample-UB} follows from \Cref{lem:continuity-in-D} and the well-known fact that $O(|\calX| / \eps^2)$ samples suffice for $\dTV(\calD, \Dhat) \le \eps$ to hold with high probability. We then focus on the one-sided convergence rate of $\CalDist_{\Dhat}(f)$ (the second parts of \Cref{thm:sample-UB,thm:sample-LB}).

\paragraph{The $O(1/\eps^3)$ upper bound.} Let $k \ge 1$ be an integer to be chosen later, and $g \in \Cal(\calD)$ be a predictor that witnesses $\CalDist_{\calD}(f)$. By the same reasoning as in \Cref{sec:overview-algorithm}, there is another predictor $\tilde g \in \Cal(\calD)$ that predicts at most $k$ different values such that $d_{\calD}(f, \tilde g) \le \CalDist_{\calD}(f) + 1/k$. By \Cref{fact:continuity-in-f}, we have $\CalDist_{\Dhat}(f) \le d_{\Dhat}(f, \tilde g) + \CalDist_{\Dhat}(\tilde g)$, so it suffices to upper bound $d_{\Dhat}(f, \tilde g)$ and $\CalDist_{\Dhat}(\tilde g)$.

The first term is easy to control: $d_{\Dhat}(f, \tilde g)$ is simply the average of $m$ independent $[0, 1]$-valued random variables with mean $d_{\calD}(f, \tilde g)$, so we have $d_{\Dhat}(f, \tilde g) \le d_{\calD}(f, \tilde g) + O(1/\sqrt{m}) \le \CalDist_{\calD}(f) + 1/k + O(1/\sqrt{m})$ with high probability.

We control the second term $\CalDist_{\Dhat}(\tilde g)$ by arguing that $\tilde g$ has a small bias at each predicted value over $\Dhat$. Suppose for simplicity that $\tilde g$ predicts each of the $k$ values exactly $m / k$ times over the sample. Then, the empirical average of the labels deviates from the prediction by $O(1/\sqrt{m/k})$ in expectation. Therefore, $\tilde g$ can be made perfectly calibrated over $\Dhat$ by changing every prediction by $O(1/\sqrt{m/k})$. The general case---that each of the $k$ values might be predicted a different number of times---can be analyzed similarly.

Therefore, it holds with high probability that $\CalDist_{\Dhat}(f) \le \CalDist_{\calD}(f) + 1/k + O(1/\sqrt{m / k})$. Setting $k = m^{1/3}$ gives an $O(m^{-1/3})$ bound on the over-estimation.\footnote{We note that a similar analysis and the same trade-off appears in a proof of Sergiu Hart~\cite{FV98,Hart25}, which shows that a horizon length of $O(1/\eps^3)$ is sufficient for achieving an $\eps$ ECE in sequential calibration.} Equivalently, a size-$O(1/\eps^3)$ sample suffices for bounding the over-estimation by $\eps$.

\paragraph{Tightness of the $1/\eps^3$ dependence.} The $1/\eps^3$ dependence is not an artifact of the upper bound proof; we prove a matching lower bound essentially via reverse engineering the analysis.  Consider the domain $\calX = [k] = \{1, 2, \ldots, k\}$ and instance $(\calD, f)$ specified by
\[
    \calD_x(i) = 1/k
\quad \text{and} \quad
    \mu_{\calD}(i) = f(i) = \frac{1}{3} + \frac{i}{3k}, ~\forall i \in [k].
\]
All the conditional probabilities lie in $[1/3, 2/3]$ and are separated by $\Omega(1/k)$. Note that $f$ is perfectly calibrated over $\calD$, so $\CalDist_{\calD}(f) = 0$. We will show that a sample of size $m = k^3$ can still lead to $\CalDist_{\Dhat}(f) = \Omega(1/k)$, i.e., an $\Omega(m^{-1/3})$ over-estimation. 

We say that element $i \in [k]$ is \emph{typical} if both of the following hold: (1) $i$ appears $\approx m / k = k^2$ times in the sample, which implies $\Dhat(i) = \Omega(1/k)$; (2) the empirical average of labels at $i$ (namely, $\mu_{\Dhat}(i)$) deviates from $f(i) = \mu_{\calD}(i)$ by $\Omega(1/k)$. By standard concentration and anti-concentration bounds, each element is typical with high probability.

Assume for simplicity that all elements are typical. We consider $\CalDist_{\Dhat}(f)$ from the partition perspective. For each part $\calX' = \{i\}$ of size $1$, the cost is given by $\Dhat(i) \cdot |f(i) - \mu_{\Dhat}(i)| = \Omega(1/k^2)$. For each part $\calX'$ with $l \ge 2$ elements, since the predictions $f(1), f(2), \ldots, f(k)$ are separated by $\Omega(1/k)$, we need to move at least $l - 1 \ge l/2$ predictions by $\Omega(1/k)$ regardless of the value of $\mu_{\Dhat}(\calX')$. The cost of $\calX'$ is then at least $\Omega(l / k^2)$. We conclude that the cost of each part is lower bounded by $\Omega(1/k^2)$ times its size, so every partition of $\calX$ has a cost of at least $k \cdot \Omega(1/k^2) = \Omega(1/k)$.

In the actual proof, we can only guarantee that a large fraction of the elements are typical. Fortunately, the argument above is robust enough to handle this case as well.

\subsection{Lower Bound for Sampled-Based Estimation}
For brevity, we assume that $\eps = \Omega(1)$. We will construct two distributions that are indistinguishable with $\ll \sqrt{|\calX|}$ samples, but lead to calibration distances that differ by $\Omega(1)$. The construction is again inspired by the example of~\cite{BGHN23} (see \Cref{sec:overview-hardness}).

Consider the domain $\calX = \{x^{-}, x^{+}\} \cup \{x_1, x_2, \ldots, x_k\}$ and instance $(\calD, f)$ defined as follows:
\begin{itemize}
    \item $\calD_x(x^{-}) = 1/4$, $\mu_{\calD}(x^{-}) = 1/2$, and $f(x^{-}) = 1/3$.
    \item $\calD_x(x^{+}) = 1/4$, $\mu_{\calD}(x^{+}) = 1/2$, and $f(x^{-}) = 2/3$.
    \item $\calD_x(x_i) = 1/(2k)$ and $f(x_i) = 1/2$.
\end{itemize}
From the transportation perspective, we have $k$ points $x_1, x_2, \ldots, x_k$ with equal mass $1/(2k)$ at $1/2$. Two additional points, each with mass $1/4$ and value $1/2$, are on the two sides of $1/2$. For the conditional probabilities (i.e., values) of $x_1, x_2, \ldots, x_k$, we consider the following two cases:
\begin{itemize}
    \item \textbf{The pure case:} $\mu_{\calD}(x_i) = 1/2$ for every $i$.
    \item \textbf{The mixed case:} Each $\mu_{\calD}(x_i)$ is independently drawn from $\Bern(1/2)$.
\end{itemize}

We first note that the two cases above are indistinguishable unless a collision occurs, i.e., some element $x_i$ appears more than once in the sample. This is because, conditioning on that there are no collisions, the labels in the sample independently follow $\Bern(1/2)$. We need a sample size of $\Omega(\sqrt{|\calX|})$ for a collision to happen.

Next, we note an $\Omega(1)$ gap in the calibration distances. In the pure case, since every point is of value $1/2$, we have to move all points to $1/2$, at a total cost of $(1/4) \cdot |1/3 - 1/2| + (1/4) \cdot |2/3 - 1/2| = 1/12$. It follows that $\CalDist_{\calD}(f) = 1/12$.

In the mixed case, we expect to have $\approx k/2$ points with value $0$ and $1$ each. Then, once we move $k/4$ $0$-valued points from $1/2$ to $1/3$ and another $k/4$ $1$-valued points to $2/3$, the configuration becomes calibrated at both $1/3$ and $2/3$. We move the remaining $k/2$ points to the average of their values, which is typically $1/2 \pm O(1/\sqrt{k})$. Therefore, we have
\begin{align*}
    \CalDist_{\calD}(f)
&\le(k/4) / (2k) \cdot |1/3 - 1/2| + (k/4) / (2k) \cdot |2/3 - 1/2| + (k/2) / (2k) \cdot O(1/\sqrt{k})\\
&=  1/24 + O(1/\sqrt{k}),
\end{align*}
establishing an $\Omega(1)$ gap between the pure and mixed cases.

%% file: discussion.tex
\section{Discussion}\label{sec:discussion}
We discuss the consequences of our results and highlight a few open problems.

\paragraph{Certifying high calibration distances.} The complexity of our $\NP$-hardness proofs arises from the lack of efficient certificates for high calibration distances: a predictor $g \in \Cal(\calD)$ with $d_{\calD}(f, g) \le \gamma$ serves as a ``proof'' of $\CalDist_{\calD}(f) \le \gamma$; yet, there seem to be no obvious proofs of $\CalDist_{\calD}(f) > \gamma$. One may hope that a proof can be obtained by carefully choosing a ``test'' function or a ``potential''---similarly to certifying a large \emph{smooth calibration error}~\cite{KF08} using a $1$-Lipschitz function that witnesses a high weighted bias.

The $\NP$-hardness results (\Cref{thm:NP-hard}) suggest that this difficulty is unavoidable in general: the same proofs show that the decision version---namely, deciding whether $\CalDist_{\calD}(f) \le \gamma$ holds---is $\NP$-complete.\footnote{The problem is in $\NP$ since a perfectly calibrated predictor gives an efficiently-verifiable proof.} This implies that the complement---deciding whether $\CalDist_{\calD}(f) > \gamma$ holds---is $\coNP$-complete and thus not in $\NP$ unless $\NP = \coNP$.

The same asymmetry arises in the estimation problem. The empirical calibration distance is unlikely to be over-estimating (\Cref{thm:sample-UB}), since the closest calibrated predictor gives a certificate that is preserved in the empirical distribution. In contrast, \Cref{thm:sample-LB} shows that there is no succinct ``proof'' of a lower bound on $\CalDist_{\calD}(f)$ that leads to a $\poly(1/\eps)$ sample complexity for bounding under-estimation by $\eps$.

\paragraph{An FPTAS?} An obvious open question on the computation side is whether a fully polynomial-time approximation scheme (FPTAS) exists. Recall that the current PTAS essentially solves a multiple knapsack problem with $O(1/\eps)$ knapsacks. One promising avenue is to prove a structural result that further prunes the search space for the partition that (approximately) witnesses $\CalDist_{\calD}(f)$. For example, it might be true that, once the elements are sorted in increasing order of predictions, there exists an efficient partition that does not involve overlaps between parts that are far apart. This would then allow the algorithm to track only at most $O(1)$ different knapsacks at a time, resulting in an FPTAS.

\paragraph{Finer-grained results for estimation and testing.} For the estimation problem, we leave a quadratic gap in the $|\calX|$-dependence of the sample complexity. Tighter bounds may be possible by drawing connection to related problems such as entropy estimation. By the type-sparsification approach (\Cref{sec:overview-algorithm}), the calibration distance only depends on the \emph{histogram} of the marginal distribution---namely, the number of elements with probability mass $\approx (1 + \eps)^{-k}$ for each $k = 0, 1, 2, \ldots$---which underlies the sublinear estimators for entropy and domain size~\cite{VV17}.

One may also consider the related problem of \emph{tolerant testing}, i.e., to distinguish the two cases $\CalDist_{\calD}(f) \le \eps_1$ and $\CalDist_{\calD}(f) \ge \eps_2$ for some $\eps_2 > \eps_1 > 0$. The $\Omega(\sqrt{|\calX|} / \eps)$ lower bound in \Cref{thm:sample-LB} applies to the case that $\eps_1, \eps_2 = \Theta(\eps)$. In contrast, the sample complexity significantly drops to $O(1/\eps^2)$ in the regime that $\eps_1 = \eps$ and $\eps_2 \gg \sqrt{\eps}$, where the problem can be reduced to estimating the smooth calibration error to error $O(\eps)$. Much is left to be understood regarding the landscape between these two cases.

%% file: algorithm.tex
\section{Exact Computation and Approximation}\label{sec:algorithm}
With respect to distribution $\calD \in \Delta(\calX \times \zo)$, the \emph{type} of $x \in \calX$ is the pair $(\calD_x(x), \mu_{\calD}(x))$. A distribution $\calD$ is \emph{$k$-type-sparse} if there are at most $k$ different types of elements. For instance, every uniform and noiseless instance $\calD$ is $2$-type-sparse, as the only possible types are $(1/|\calX|, 0)$ and $(1/|\calX|, 1)$.

\subsection{Algorithm for the Type-Sparse Case}
The distance from calibration can be efficiently computed on type-sparse instances.

\begin{lemma}\label{lem:sparse-case}
    Given a $k$-type-sparse distribution $\calD$ and predictor $f: \calX \to [0, 1]$, $\CalDist_{\calD}(f)$ can be computed exactly in $O\left(\poly(k) \cdot |\calX|^{2k}\right)$ time.
\end{lemma}

Let $\calD$ be a $k$-type-sparse distribution. For each $i \in [k]$, let $n_i$ be the number of elements of the $i$th type, and re-label these $n_i$ elements as $x_{i, 1}, x_{i, 2}, \ldots, x_{i, n_i}$ such that
\[
    f(x_{i, 1}) \le f(x_{i, 2}) \le \cdots \le f(x_{i, n_i}).
\]
We focus on predictors that preserve the ordering of $f$ within each type.

\begin{definition}[Type-monotonicity]
    A predictor $g: \calX \to [0, 1]$ is type-monotone if $g(x_{i,j}) \le g(x_{i, j+1})$ holds for every $i \in [k]$ and $j \in \{1, 2, \ldots, n_i - 1\}$.
\end{definition}

The following lemma states that $\CalDist_{\calD}(f)$ is witnessed by a type-monotone predictor. The lemma is a simple consequence of optimal transport on the line, and its proof is deferred to \Cref{app:algorithm}.

\begin{lemma}\label{lem:type-monotonicity}
    For any distribution $\calD$ and predictor $f$, there exists a type-monotone predictor $g \in \Cal(\calD)$ such that $\CalDist_{\calD}(f) = d_{\calD}(f, g)$.
\end{lemma}

Now, we are ready to prove \Cref{lem:sparse-case}.

\begin{proof}[Proof of \Cref{lem:sparse-case}]
    By \Cref{lem:type-monotonicity}, $\CalDist_{\calD}(f)$ is witnessed by a type-monotone predictor, which induces (by \Cref{fact:CalDist-vs-partition}) a \emph{contiguous} partition $\calP = \{\calX_1, \calX_2, \ldots, \calX_l\}$ in which every part $\calX_t$ is of form
    \[
        \bigcup_{i = 1}^{k}\{x_{i, j}: s_{i, t-1} < j \le s_{i,t}\},
    \]
    where $0 = s_{i, 0} \le s_{i, 1} \le \cdots \le s_{i,l} = n_i$ for each $i \in [k]$. In words, the $n_i$ elements of type~$i$ are divided into $l$ (possibly empty) contiguous subsets, and each of the $l$ parts in $\calP$ is formed by combining the corresponding subsets of the $k$ types.
    
    Therefore, $\CalDist_{\calD}(f)$ is given by the minimum of $\cost_{\calD, f}(\calP)$ over all contiguous partitions, which can be computed via dynamic programming as follows. For integers $m_1 \in [0, n_1], m_2 \in [0, n_2], \ldots, m_k \in [0, n_k]$, let $\opt(m_1, m_2, \ldots, m_k)$ denote the minimum total cost among all contiguous partitions of $\bigcup_{i=1}^{k}\{x_{i, j}: 1 \le j \le m_i\}$. By definition, every such contiguous partition must contain a part of form
    \[
        \calX(m', m) \coloneqq \bigcup_{i=1}^{k}\{x_{i,j}: m'_i < j \le m_i\}
    \]
    for some vector $m' \ne m$ that is upper bounded by $m$ entrywise. This gives the recurrence relation
    \[
        \opt(m)
    =   \min_{m'}\left[\opt(m') + \cost_{\calD, f}(\calX(m', m))\right].
    \]

    The dynamic programming algorithm involves $O(|\calX|^k)$ values, each of which can be computed by enumerating $O(|\calX|^k)$ possible choices of $m'$. Computing $\cost_{\calD, f}(\calX(m', m))$ in a straightforward way takes $O(|\calX|)$ time, since we need to compute
    \[
        \mu(\calX(m', m)) = \frac{\sum_{i=1}^{k}\sum_{j=m'_i + 1}^{m_i}\calD_x(x_{i,j}) \cdot \mu(x_{i,j})}{\sum_{i=1}^{k}\sum_{j=m'_i + 1}^{m_i}\calD_x(x_{i,j})}
    \]
    and then
    \[
        \sum_{i=1}^{k}\sum_{j=m'_i+1}^{m_i}\calD_x(x_{i,j}) \cdot |f(x_{i,j}) - \mu(\calX(m', m))|.
    \]
    Each double summation takes $O(|\calX|)$ time to compute.

    This $O(|\calX|)$ runtime for evaluating $\cost_{\calD, f}(\calX(m', m))$ can be improved to $O(k)$ via proper pre-processing. Fix $i \in [k]$ and recall that $f(x_{i,j})$ is non-decreasing in $j$. If we know the maximum index $j_0$ such that $f(x_{i,j_0}) < \mu(\calX(m', m))$, we can compute the sum
    \[
        \sum_{j=m'_i+1}^{m_i}\calD_x(x_{i,j}) \cdot \left|f(x_{i,j}) - \mu(\calX(m', m))\right|
    \]
    from $\mu(\calX(m', m))$ and partial sums of $\calD_x(x_{i,j})$ and $\calD_x(x_{i,j}) \cdot f(x_{i,j})$ in $O(1)$ time. The entire cost can then be computed in $O(k)$ time.
    
    Since both $\calD_x(x_{i,j})$ and $\mu(x_{i,j})$ depend only on $i$ (and not $j$), we can re-write $\mu(\calX(m', m))$ as
    \[
        \frac{\sum_{i=1}^{k}(m_i - m'_i)\cdot \calD_x(x_{i,1}) \cdot \mu(x_{i,1})}{\sum_{i=1}^{k}(m_i - m'_i)\cdot \calD_x(x_{i,1})},
    \]
    which is a function of $m - m'$ and thus takes at most $O(|\calX|^k)$ different values. Therefore, we can compute the maximum index $j_0$ such that $f(x_{i,j_0}) < \mu(\calX(m', m))$ for every type $i \in [k]$ and every possible value of $\mu(\calX(m', m))$ in time $O_k(|\calX|^k \cdot |\calX|) = O_k(|\calX|^{k+1})$, which will be dominated by the $O_k(|\calX|^{2k})$ runtime bound of dynamic programming. Here, the $O_k$ notation hides a $\poly(k)$ factor that accounts for the overhead in implementing $k$ levels of nested loops and accessing entries of a $k$-dimensional array, which depends on the exact model of computation. The final runtime of the algorithm is then $O(\poly(k) \cdot |\calX|^{2k})$.
\end{proof}

\subsection{A Type-Sparsification Lemma}
Next, we show that every distribution $\calD$ can be transformed into a type-sparse one that is close in the TV distance.

\begin{lemma}\label{lem:type-sparsification}
    For any $\calD \in \Delta(\calX \times \zo)$ and $\eps > 0$, there is an $O(\log(|\calX| / \eps) / \eps^2)$-type-sparse distribution $\calD'$ such that $\dTV(\calD, \calD') \le \eps$. The sparsity is improved to $O(1/\eps)$ in the uniform case.
\end{lemma}

\Cref{lem:sparse-case,lem:type-sparsification} together imply part of \Cref{thm:runtime}, with a slower runtime for the general case.
\begin{corollary}\label{cor:runtime-from-type-sparsity}
    Given $\calD$ and $f$, $\CalDist_{\calD}(f)$ can be computed:
    \begin{itemize}
        \item Exactly and in $O(|\calX|^4)$ time if $\calD$ is both uniform and noiseless.
        \item Up to $\eps$ error and in $|\calX|^{O(1/\eps)}$ time in the uniform case.
        \item Up to $\eps$ error and in $|\calX|^{O(\log(|\calX|/\eps)/\eps^2)}$ time in the general case.
    \end{itemize}
\end{corollary}

\begin{proof}[Proof of \Cref{lem:type-sparsification}]
    We first prove the uniform case by rounding the conditional probabilities. We then handle the more general case with an additional rounding of the marginal probabilities.

    \paragraph{The uniform case.} Let $\calD'$ be the distribution with the same $\calX$-marginal as $\calD$ and
    \[
        \mu_{\calD'}(x) = \left\lfloor\mu_{\calD}(x) / \eps\right\rfloor \cdot \eps
    \]
    for each $x \in \calX$. Then, $\calD'$ is $O(1/\eps)$-type-sparse and
    \[
        \dTV(\calD, \calD')
    =   \sum_{x \in \calX}\calD_x(x) \cdot |\mu_{\calD}(x) - \mu_{\calD'}(x)|
    \le \sum_{x \in \calX}\calD_x(x) \cdot \eps
    =   \eps.
    \]

    \paragraph{The general case.} The same construction allows us to transform a general $\calD$ into $\calD'$ such that $\calD_x = \calD'_x$, $\dTV(\calD, \calD') \le \eps / 2$, and $|\{\mu_{\calD'}(x): x \in \calX\}| = O(1/\eps)$.

    Next, we transform $\calD'$ into another distribution $\calD''$ by rounding the marginal probabilities while keeping the conditional ones. Let $\delta \in (0, 1/2)$ be a parameter to be chosen later. For each $x \in \calX$, define
    \[
        a(x) \coloneqq \begin{cases}
            0, & \calD'_x(x) \le \delta / |\calX|,\\
            (1 + \delta)^{\lfloor \log_{1+\delta}\calD'_x(x)\rfloor}, & \calD'_x(x) > \delta / |\calX|.
        \end{cases}
    \]
    Note that $a(x) \le \calD'_x(x)$ holds for all $x \in \calX$. In the former case, $|a(x) - \calD'_x(x)| \le \delta / |\calX|$. In the latter case, we have $a(x) \ge \calD'_x(x) / (1 + \delta)$, which implies
    \[
        \left|a(x) - \calD'_x(x)\right| \le \calD'_x(x) \cdot \left(1 - \frac{1}{1 + \delta}\right)
    \le \delta \cdot \calD'_x(x).
    \]
    Therefore,
    \begin{equation}\label{eq:a-vs-Dp}
        \sum_{x \in \calX}|a(x) - \calD'_x(x)|
    \le \sum_{x \in \calX}\frac{\delta}{|\calX|} + \delta \cdot \sum_{x \in \calX}\calD'_x(x)
    =   2\delta.
    \end{equation}

    Let $S \coloneqq \sum_{x \in \calX}a(x)$. We have $S \ge \sum_{x \in \calX}\calD'_x(x) - 2\delta = 1 - 2\delta$. Let $\calD''$ be the distribution with $\calD''_x(x) = a(x) / S$ and $\mu_{\calD''}(x) = \mu_{\calD'}(x)$ for every $x \in \calX$. Since $S \in [1 - 2\delta, 1]$, we have
    \begin{equation}\label{eq:a-vs-Dpp}
        \sum_{x \in \calX}|a(x) - \calD''_x(x)|
    =   \sum_{x \in \calX}\calD''_x(x) \cdot |S - 1|
    =   1 - S
    \le 2\delta.
    \end{equation}

    Since $\calD'$ and $\calD''$ have the same conditional probabilities,
    \begin{align*}
        \dTV(\calD', \calD'')
    &=   \dTV(\calD'_x, \calD''_x)
    =   \frac{1}{2}\sum_{x \in \calX}|\calD'_x(x) - \calD''_x(x)|\\
    &\le\frac{1}{2}\left(\sum_{x \in \calX}|a(x) - \calD'_x(x)| + \sum_{x \in \calX}|a(x) - \calD''_x(x)|\right)\\
    &\le \frac{1}{2}\cdot (2\delta + 2\delta)
    =   2\delta. \tag{\Cref{eq:a-vs-Dp,eq:a-vs-Dpp}}
    \end{align*}
    Setting $\delta = \eps / 4$ leads to $\dTV(\calD', \calD'') \le \eps / 2$ and
    \[
        \dTV(\calD, \calD'')
    \le \dTV(\calD, \calD') + \dTV(\calD', \calD'')
    \le \eps.
    \]

    It remains to bound the type-sparsity of $\calD''$. Recall that $\{\mu_{\calD''}(x): x \in \calX\}$ is of size $O(1/\eps)$. Moreover, $\{\calD''_x(x): x \in \calX\}$ is a subset of
    \[
        \{0\} \cup \{(1 + \delta)^k: k \in \{{\lfloor \log_{1+\delta}(\delta / |\calX|)\rfloor}, \ldots, -1, 0\}\},
    \]
    which is of size $O(\log_{1+\delta}(|\calX| / \delta)) = O(\log(|\calX| / \eps) / \eps)$. Thus, $\calD''$ is $O(\log(|\calX| / \eps) / \eps^2)$-type-sparse.
\end{proof}

\subsection{A PTAS via Prediction-Sparsity}
In this section, we improve the QPTAS for the general case in \Cref{cor:runtime-from-type-sparsity} to a PTAS by using a slightly different approach. Instead of grouping the domain elements into a few types, we restrict the ``target'' predictor in the definition of the calibration distance. Specifically, it suffices to consider $g \in \Cal(\calD)$ that maps to at most $O(1/\eps)$ different values.

\begin{lemma}\label{lem:prediction-sparsification}
    For every distribution $\calD$, predictor $f$, and integer $k \ge 1$, there exists $g \in \Cal(\calD)$ such that: (1) $d_{\calD}(f, g) \le \CalDist_{\calD}(f) + 1/k$; (2) $g$ predicts at most one value in each of $\calI_1, \calI_2, \ldots, \calI_k$, where $\calI_i = [(i - 1) / k, i / k)$ for $i \in [k - 1]$ and $\calI_k = [(k-1)/k, 1]$.
\end{lemma}

\begin{proof}[Proof of \Cref{lem:prediction-sparsification}]
    By definition, there exists $\tilde f \in \Cal(\calD)$ such that $d_{\calD}(f, \tilde f) = \CalDist_{\calD}(f)$. For each $i \in [k]$, define
    \[
        \calX_i \coloneqq \tilde f^{-1}(\calI_i) = \{x \in \calX: \tilde f(x) \in \calI_i\}.
    \]
    Clearly, $\{\calX_1, \calX_2, \ldots, \calX_k\}$ is a partition of $\calX$ and induces a perfectly calibrated predictor
    \[
        g(x) \coloneqq \mu(\calX_i), ~\forall i \in [k], x \in \calX_i.
    \]
    Note that each $\mu(\calX_i)$ is a convex combination of values in the length-$(1/k)$ interval $\calI_i$, so we have $d_{\calD}(\tilde f, g) \le \|\tilde f - g\|_{\infty} \le 1 / k$. It follows that
    \[
        d_{\calD}(f, g)
    \le d_{\calD}(f, \tilde f) + d_{\calD}(\tilde f, g)
    \le \CalDist_{\calD}(f) + 1/k.
    \]
\end{proof}

Another ingredient of the algorithm is to round the probabilities to integer multiples of $\Theta(\eps / |\calX|)$. We say that a distribution is $M$-discrete, if the probability mass on each element is an integer multiple of $1/M$.

\begin{lemma}\label{lem:discretization}
    For every distribution $\calD \in \Delta(\calX \times \zo)$ and $\eps > 0$, there exists a distribution $\calD'$ such that $\dTV(\calD, \calD') \le \eps$ and $\calD'$ is $M$-discrete for $M = O(|\calX| / \eps)$.
\end{lemma}

\begin{proof}[Proof of \Cref{lem:discretization}]
    Let $N$ be a positive integer to be chosen later. For each $(x, y) \in \calX \times \zo$, let
    \[
        a(x, y) \coloneqq \frac{\lfloor \calD(x, y) \cdot N\rfloor}{N} \in [\calD(x, y) - 1/N, \calD(x, y)].
    \]
    We have
    \[
        \sum_{(x, y) \in \calX \times \zo}\left|a(x, y) - \calD(x, y)\right|
    \le \frac{2|\calX|}{N}.
    \]
    
    Let $S \coloneqq \sum_{(x,y) \in \calX \times \zo}a(x, y)$ and note that $S \in \left[1 - \frac{2|\calX|}{N}, 1\right]$. Consider the distribution $\calD'$ defined as $\calD'(x, y) = a(x, y) / S$. We have
    \[
        \sum_{(x, y) \in \calX \times \zo}\left|a(x, y) - \calD'(x, y)\right|
    =   \sum_{(x, y) \in \calX \times \zo}\calD'(x, y) \cdot |1 - S|
    \le \frac{2|\calX|}{N}.
    \]
    It follows that
    \[
        \dTV(\calD, \calD')
    =   \frac{1}{2}\sum_{(x, y) \in \calX \times \zo}\left|\calD(x, y) - \calD'(x, y)\right|
    \le \frac{2|\calX|}{N}.
    \]

    Therefore, there exists $N = O(|\calX| / \eps)$ such that $\dTV(\calD, \calD') \le \eps$. Note that $\calD'$ can be equivalently written as
    \[
        \calD'(x, y)
    =   \frac{\lfloor \calD(x, y) \cdot N\rfloor}{\sum_{(x', y') \in \calX \times \zo}\lfloor \calD(x', y') \cdot N\rfloor},
    \]
    and is thus $M$-discrete for $M = \sum_{(x, y) \in \calX \times \zo}\lfloor \calD(x, y) \cdot N\rfloor \le N = O(|\calX| / \eps)$.
\end{proof}

Now, we give the PTAS for the general case in \Cref{thm:runtime}.

\begin{proof}[Proof of \Cref{thm:runtime}]
    The $O(|\calX|^4)$-time exact algorithm for the uniform and noiseless case and the PTAS for the uniform case follow from \Cref{cor:runtime-from-type-sparsity}. For the general case, we first apply \Cref{lem:discretization} to obtain an $M$-discrete distribution $\calD'$ for some $M = O(|\calX| / \eps)$ such that $\dTV(\calD, \calD') \le \eps / 15$. By \Cref{lem:continuity-in-D}, $|\CalDist_{\calD}(f) - \CalDist_{\calD'}(f)| \le \eps / 3$, so it suffices to approximate $\CalDist_{\calD'}(f)$ up to an error of  $2\eps / 3$.

    \paragraph{Special partitions.} By \Cref{lem:prediction-sparsification}, for some $k = O(1/\eps)$, there exists a predictor $g \in \Cal(\calD')$ that predicts at most one value from each of $\calI_1, \calI_2, \ldots, \calI_k$ and satisfies $d_{\calD'}(f, g) \le \CalDist_{\calD'}(f) + \eps / 3$.
    
    We define a sub-class of partitions of $\calX$ that induce predictors similar to $g$. A partition $\calP$ of $\calX$ is \emph{$k$-special} if $\calP$ can be written as $\{\calX_1, \calX_2, \ldots, \calX_k\}$---where we allow some parts to be empty---such that we have either $\calX_i = \emptyset$ or $\mu_{\calD'}(\calX_i) \in \calI_i$ for each $i \in [k]$.
    
    Clearly, every $k$-special partition $\calP$ satisfies $\cost_{\calD', f}(\calP) \ge \CalDist_{\calD'}(f)$. Moreover, there exists a $k$-special partition $\calP$ (specifically, the partition induced by $g$) with $\cost_{\calD', f}(\calP) \le \CalDist_{\calD'}(f) + \eps / 3$.

    \paragraph{Cost upper bound.} Next, we introduce an upper bound for the partition cost that is easier to track in dynamic programming. For each $k$-special partition $\calP = \{\calX_1, \calX_2, \ldots, \calX_k\}$, we define
    \[
        \uppercost_{\calD', f}(\calX_i)
    \coloneqq \sum_{x \in \calX_i}\calD'_x(x) \cdot \left[\left|f(x) - \frac{i - 1/2}{k}\right| + \frac{1}{2k}\right]
    \]
    and $\uppercost_{\calD', f}(\calP) = \sum_{i=1}^{k}\uppercost_{\calD', f}(\calX_i)$.

    Since a $k$-special partition satisfies $\mu_{\calD'}(\calX_i) \in \calI_i \subseteq [(i-1)/k, i/k]$, it holds for all $x \in \calX_i$ that
    \[
        |f(x) - \mu_{\calD'}(\calX_i)| \le \left|f(x) - \frac{i - 1/2}{k}\right| + \frac{1}{2k} \le |f(x) - \mu_{\calD'}(\calX_i)| + \frac{1}{k}.
    \]
    Summing over all $i \in [k]$ and $x \in \calX_i$ gives
    \[
        \cost_{\calD', f}(\calP) \le \uppercost_{\calD', f}(\calP) \le \cost_{\calD', f}(\calP) + \frac{1}{k}.
    \]
    It follows that every $k$-special partition $\calP$ satisfies $\uppercost_{\calD', f}(\calP) \ge \cost_{\calD', f}(\calP) \ge \CalDist_{\calD'}(f)$. Moreover, the $k$-special partition $\calP$ induced by $g$ gives $\uppercost_{\calD', f}(\calP) \le \cost_{\calD', f}(\calP) + 1/k \le \CalDist_{\calD'}(f) + 2\eps / 3$, where the last step holds for a sufficiently large $k = O(1/\eps)$.

    \paragraph{Minimize $\uppercost$ via dynamic programming.} It remains to compute the minimum of $\cost_{\calD', f}(\calP)$ over all $k$-special partitions. This can be done in $O(|\calX| \cdot (M + 1)^k \cdot k) = (|\calX| / \eps)^{O(1/\eps)}$ time via dynamic programming.

    Write $\calX = \{x_1, x_2, \ldots, x_{|\calX|}\}$. For every $i \in [|\calX|]$ and $2k$-tuple
    \[
        p = (p_{1,0}, p_{1,1}, p_{2,0}, p_{2,1}, \ldots, p_{k,0}, p_{k,1}) \in \{0, 1/M, 2/M, \ldots, 1-1/M, 1\}^{2k},
    \]
    let $\opt(i, p)$ denote the minimum value of $\uppercost_{\calD', f}(\calP)$ over all partitions $\calP = \{\calX_1, \calX_2, \ldots, \calX_k\}$ of $\{x_1, x_2, \ldots, x_i\}$ subject to
    \[
        p_{j,b} = \sum_{x \in \calX_j}\calD'(x, b)
    \]
    for every $j \in [k]$ and $b \in \zo$. In other words, $i$ tracks the number of elements that have been assigned to the $k$ parts, while $p$ tracks and total $0$- and $1$-labeled probability masses in each part.

    We have the recurrence relation
    \[
        \opt(i, p)
    =   \min_{j \in [k]}\left[\opt\left(i - 1, p^{(j)}\right) + \calD'_x(x_i) \cdot \left[\left|f(x_i) - \frac{j - 1/2}{k}\right| + \frac{1}{2k}\right]\right],
    \]
    where each $p^{(j)}$ is obtained from $p$ by subtracting $\calD'(x_i, 0)$ from $p_{j, 0}$ and subtracting $\calD'(x_i, 1)$ from $p_{j, 1}$. Moreover, the minimum of $\uppercost_{\calD', f}(\calP)$ over all $k$-special partitions is the minimum of $\opt(|\calX|, p)$ over all $2k$-tuples that satisfy either $p_{i,0} = p_{i,1} = 0$ or $\frac{p_{i, 1}}{p_{i, 0} + p_{i, 1}} \in \calI_i$ for every $i \in [k]$.

    Since there are at most $|\calX| \cdot (M + 1)^k$ values to be computed, while computing each value takes $O(k)$ time, the dynamic programming algorithm can be implemented in $(|\calX| / \eps)^{O(1/\eps)}$ time.
\end{proof}

%% file: hardness.tex
\section{Hardness of Exact Computation}\label{sec:hardness}
We prove \Cref{thm:NP-hard} by reducing from the Subset Sum Problem (SSP) and its variants.

\begin{definition}[Subset Sum Problem]\label{def:SSP}
    Given $n$ positive integers $a_1, a_2, \ldots, a_n$ and $\theta \le \frac{1}{2}\sum_{i=1}^{n}a_i$, output ``Yes'' if there exists a set $\calI \subseteq [n]$ such that $\sum_{i \in \calI}a_i = \theta$ and output ``No'' otherwise.
\end{definition}

Let $S \coloneqq \sum_{i=1}^{n}a_i$. The restriction $\theta \le S / 2$ is without loss of generality: If $\theta > S$, the correct answer is clearly ``No''. If $\theta \in (S / 2, S]$, replacing it with $S - \theta$ does not change the answer.

\subsection{The Noiseless Case}
Given an SSP instance $(a_1, a_2, \ldots, a_n, \theta)$, let
\[
    \pstar \coloneqq \frac{1}{4} \cdot \frac{\theta}{S}
\quad \text{and} \quad
    \alpha \coloneqq \frac{\pstar}{1 - 2\pstar}.
\]
Since $0 < \theta \le S / 2$, we have $\pstar \in (0, 1/8]$ and $\alpha \in (0, 1/6]$.

We use the SSP instance to construct an instance of computing the calibration distance on a noiseless (but non-uniform) distribution.

\begin{definition}\label{def:reduction-noiseless}
    Given $(a_1, a_2, \ldots, a_n, \theta)$, we construct a distribution $\calD \in \Delta(\calX \times \zo)$ and predictor $f: \calX \to [0, 1]$ as follows: The domain $\calX \coloneqq \{x_0, x_1, x_{1+}\} \cup X_{0^+}$ contains the following four groups of elements:
    \begin{itemize}
        \item Group~$0$: Element $x_0$ with $f(x_0) = 1/2$, $\calD_x(x_0) = 1/4 - \pstar/2$, and $\mu(x_0) = 0$.
        \item Group~$1$: Element $x_1$ with $f(x_1) = 1/2$, $\calD_x(x_1) = 1/4 + \pstar/2$, and $\mu(x_1) = 1$.
        \item Group~$0^{+}$: A set of $n$ elements $X_{0^+} \coloneqq \{x_{0^+, 1}, x_{0^+, 2}, \ldots, x_{0^+, n}\}$ with $f(x_{0^+, i}) = 1/2 + \alpha$, $\calD_x(x_{0^+, i}) = \frac{1}{4} \cdot \frac{a_i}{S}$, and $\mu(x_{0^+, i}) = 0$ for every $i \in [n]$.
        \item Group~$1^{+}$: Element $x_{1^+}$ with $f(x_{1^+}) = 1/2 + \alpha$, $\calD_x(x_{1^+}) = 1/4$, and $\mu(x_{1^+}) = 1$.
    \end{itemize}
\end{definition}

The bit in the group indicates whether the (noiseless) label is $0$ or $1$ given the element, while the presence of ``$+$'' indicates that the prediction made by $f$ is higher than $1/2$ by $\alpha$.

\begin{proof}[Proof of \Cref{thm:NP-hard}, the first part]
Consider an SSP instance $(a_1, a_2, \ldots, a_n, \theta)$ and the instance $(\calD, f)$ resulting from \Cref{def:reduction-noiseless}. We will show that $\CalDist_{\calD}(f) \le \alpha \pstar$ if and only if the SSP instance is a ``Yes'' instance. As a result, SSP is polynomial-time reducible to the exact computation of the distance from calibration in the noiseless case.

\paragraph{Analysis of ``Yes'' instances.} Suppose that there exists $\calI \subseteq [n]$ such that $\sum_{i \in \calI}a_i = \theta$. Then, $\{x_{0^+, i}: i \in \calI\}$ is a subset of $X_{0^+}$ with a total probability mass of
\[
    \calD_x(\{x_{0^+, i}: i \in \calI\})
=   \sum_{i \in \calI}\frac{1}{4} \cdot \frac{a_i}{S} = \frac{1}{4} \cdot \frac{\theta}{S} = \pstar.
\]
Then, the alternative predictor
\[
    g(x) = \begin{cases}
        1/2, & x \in \{x_0, x_1\} \cup \{x_{0^+, i}: i \in \calI\},\\
        1/2 + \alpha, & x \in \{x_{1^+}\} \cup \{x_{0^+, i}: i \in \overline{\calI}\}
    \end{cases}
\]
is perfectly calibrated and satisfies $d_{\calD}(f, g) = \alpha \pstar$. Thus, $\CalDist_{\calD}(f) \le \alpha \pstar$.

\paragraph{Analysis of ``No'' instances.} It remains to show that every ``No'' instance leads to $\CalDist_{\calD}(f) > \alpha \pstar$. To this end, we suppose towards a contradiction that some $g \in \Cal(\calD)$ satisfies $d_{\calD}(f, g) \le \alpha \pstar$. By \Cref{fact:CalDist-vs-partition}, $g$ naturally induces a partition of $\calX$.

We first note that $x_0$ and $x_{1^+}$ cannot be in the same part. If $g$ maps $x_0$ and $x_{1^+}$ to the same value (denoted by $q$), recalling $\calD_x(x_0) = 1/4 - \pstar/2$, $\calD_x(x_{1^+}) = 1/4$, $f(x_0) = 1/2$, and $f(x_{1^+}) = 1/2 + \alpha$, we have
\begin{align*}
    d_{\calD}(f, g)
&\ge\calD_x(x_0) \cdot |f(x_0) - g(x_0)| + \calD_x(x_{1^+}) \cdot |f(x_{1^+}) - g(x_{1^+})|\\
&\ge(1/4 - \pstar/2) \cdot \left(|1/2 - q| + |1/2 + \alpha - q|\right)\\
&\ge(1/4 - \pstar/2) \cdot \alpha, \tag{triangle inequality}
\end{align*}
while the last line is strictly higher than $\alpha \pstar$ since $\alpha > 0$ and $\pstar \in (0, 1/8]$.

Next, we observe that $x_0$ and $x_1$ must be in the same part. Otherwise, $x_0$ would be in the same part with a subset of $X_{0^+}$ (and no other elements). This forces $g(x_0) = 0$ and thus 
\[
    d_{\calD}(f, g)
\ge \calD_x(x_0) \cdot |f(x_0) - g(x_0)|
=   (1/4 - \pstar/2) \cdot (1/2)
>   \frac{\pstar}{1 - 2\pstar} \cdot \pstar
=   \alpha \pstar,
\]
where the third step holds for all $\pstar \in (0, 1/8]$.

Therefore, the partition induced by $g$ must be of the following form:
\[
    \{\{x_0, x_1\} \cup X_1, \{x_{1^+}\} \cup X_2, X_3\},
\]
where $\{X_1, X_2, X_3\}$ is a partition of $X_{0^+}$. For each $i \in \{1, 2, 3\}$, let $p_i \coloneqq \calD_x(X_i)$ be the total probability mass of $X_i$ in the $\calX$-marginal. Note that $p_1 + p_2 + p_3 = 1/4$. Also, we must have $p_1 \ne \pstar$; otherwise, the elements in $X_1$ would correspond to a solution to the SSP instance.

We express the distance $d_{\calD}(f, g)$ in terms of $p_1$, $p_2$, and $p_3$:
\begin{itemize}
    \item The contribution from $\{x_0, x_1\} \cup X_1$ is
    \begin{align*}
        C_1
    &\coloneqq \frac{1}{2}\cdot\left|\frac{1}{2} - \frac{1/4 + \pstar / 2}{1/2 + p_1}\right| + p_1\cdot\left|\frac{1}{2} + \alpha - \frac{1/4 + \pstar / 2}{1/2 + p_1}\right|\\
    &=  \frac{|p_1 - \pstar|}{2 + 4p_1} + p_1\cdot\left|\frac{1}{2} + \alpha - \frac{1/4 + \pstar / 2}{1/2 + p_1}\right|.
    \end{align*}
    Since we always have
    \[
        \frac{1}{2} + \alpha
    =   \frac{1}{2} + \frac{\pstar}{1 - 2\pstar}
    \ge \frac{1}{2} + \pstar
    =   \frac{1/4 + \pstar / 2}{1/2}
    \ge \frac{1/4 + \pstar / 2}{1/2 + p_1},
    \]
    the second absolute value can be removed, i.e.,
    \[
        C_1 = \frac{|p_1 - \pstar|}{2 + 4p_1} + p_1\cdot\left(\frac{1}{2} + \alpha - \frac{1/4 + \pstar / 2}{1/2 + p_1}\right).
    \]
    \item The contribution from $\{x_{1^+}\} \cup X_2$ is
    \[
        C_2 \coloneqq (p_2 + 1/4) \cdot \left|(1/2 + \alpha) - \frac{1/4}{p_2 + 1/4}\right|
    =   \left|(1/2 + \alpha)p_2 - \frac{1/2 - \alpha}{4}\right|.
    \]
    \item The contribution from $X_3$ is simply $C_3 \coloneqq (1/2 + \alpha)p_3$.
\end{itemize}

We fix the value of $p_1$ and lower bound the sum $C_2 + C_3$ subject to $p_2 + p_3 = 1/4 - p_1$:
\begin{align*}
    C_2 + C_3
&=  \left|(1/2 + \alpha)p_2 - \frac{1/2 - \alpha}{4}\right| + (1/2 + \alpha)p_3\\
&\ge \left|(1/2 + \alpha)p_2 - \frac{1/2 - \alpha}{4} + (1/2 + \alpha)p_3\right| \tag{triangle inequality}\\
&=  \left|(1/2 + \alpha)(1/4 - p_1) - \frac{1/2 - \alpha}{4}\right| \tag{$p_2 + p_3 = 1/4 - p_1$}\\
&=  \left|\frac{\alpha}{2} - (1/2 + \alpha)p_1\right|.
\end{align*}

Then, $C_1 + C_2 + C_3$ is lower bounded by
\[
    \frac{|p_1 - \pstar|}{2 + 4p_1} + p_1\cdot\left(\frac{1}{2} + \alpha - \frac{1/4 + \pstar / 2}{1/2 + p_1}\right) + \left|\frac{\alpha}{2} - (1/2 + \alpha)p_1\right|.
\]
Plugging $\alpha = \frac{\pstar}{1 - 2\pstar}$ into the above gives a simplified expression of
\[
    \frac{|p_1 - \pstar|}{2 + 4p_1} + \left(\frac{p_1}{2 - 4\pstar} - \frac{1 + 2\pstar}{2 + 4p_1} \cdot p_1\right) + \frac{|p_1 - \pstar|}{2 - 4\pstar}.
\]

On $[0, 1/4] \setminus \{\pstar\}$, the function
\[
    \phi(x) \coloneqq \frac{|x - \pstar|}{2 + 4x} + \left(\frac{x}{2 - 4\pstar} - \frac{1 + 2\pstar}{2 + 4x} \cdot x\right) + \frac{|x - \pstar|}{2 - 4\pstar},
\]
has a derivative of
\[
    \phi'(x) = \frac{1 + \sgn(x - \pstar)}{2 - 4\pstar} + \frac{\sgn(x - \pstar) - (1 + 2\pstar)}{2 + 4x} + \frac{(1 + 2\pstar)x - |x - \pstar|}{(1 + 2x)^2}.
\]
When $x > \pstar$, the derivative further simplifies into
\[
    \phi'(x) = \frac{1}{1 - 2\pstar} - \frac{\pstar}{1 + 2x} + \frac{2\pstar x + \pstar}{(1 + 2x)^2}
=   \frac{1}{1 - 2\pstar} > 0.
\]
When $x < \pstar$, the derivative reduces to
\[
    \phi'(x) = -\frac{1 + \pstar}{1 + 2x} + \frac{x + 2\pstar x + x - \pstar}{(1 + 2x)^2}
=   -\frac{1 + 2\pstar}{(1 + 2x)^2}
<   0.
\]

Therefore, $\phi(x)$ is uniquely minimized at $x = \pstar$, where
\[
    \phi(x)
=   \frac{\pstar}{2 - 4\pstar} - \frac{1}{2}\pstar
=   \pstar \cdot \frac{\pstar}{1 - 2\pstar}
=   \alpha \pstar.
\]
Recall that $d_{\calD}(f, g) \ge \phi(p_1)$ while $p_1 \ne \pstar$. Therefore, we must have $d_{\calD}(f, g) > \alpha\pstar$. This leads to a contradiction and completes the proof.
\end{proof}

\subsection{The Uniform Case: Overview}
For the second part of \Cref{thm:NP-hard}, where $\calD$ is uniform but not noiseless, we reduce from the following ``balanced'' variant of the SSP.

\begin{definition}[Balanced SSP]\label{def:balanced-SSP}
    Given $n = 2k$ positive integers $a_1, a_2, \ldots, a_n$ that satisfy
    \[
        \frac{\max_{i \in [n]}a_i}{\min_{i \in [n]}a_i} \le 1 + \frac{1}{100k},
    \]
    the goal is to decide whether there is a size-$k$ set $\calI \subseteq [n]$ such that $\sum_{i \in \calI}a_i = \frac{1}{2}\sum_{i=1}^{n}a_i$.
\end{definition}

Balanced SSP is $\NP$-hard by a straightforward reduction from the partition problem (deciding whether a multiset of numbers can be split into two equal-sum parts). The condition on $\frac{\max_i a_i}{\min_i a_i}$ can be guaranteed by adding a (moderately) large number to all numbers in the instance. For completeness, we prove the following lemma in Appendix~\ref{app:hardness}.

\begin{lemma}\label{lem:balanced-SSP-NP-hard}
    Balanced SSP is $\NP$-hard.
\end{lemma}

Next, we introduce the reduction from Balanced SSP to computing the distance from calibration on a uniform distribution. We set
\[
    S \coloneqq \sum_{i=1}^{n}a_i,
\quad\quad
    \eps_i \coloneqq \frac{a_i}{S},
\quad\quad
    \eps \coloneqq \frac{1}{6k},
\quad\quad \text{and} \quad\quad
    \delta = \frac{1}{4k}.
\]
Note that the restrictions of the Balanced SSP implies that $\eps_i \approx \frac{1}{2k}$ up to a factor of $1 \pm O(1/k)$.

\begin{definition}\label{def:reduction-uniform}
    Given $(a_1, a_2, \ldots, a_n)$, we construct distribution $\calD \in \Delta(\calX \times \zo)$ and predictor $f: \calX \to [0, 1]$ as follows: The domain
    \[
        \calX
    \coloneqq \{x_i: i \in [n]\} \cup \{x'_i: i \in [n]\} \cup \{x''_i: i \in [n]\}
    \]
    is of size $3n = 6k$, and contains the following three groups of elements:
    \begin{itemize}
        \item Elements $x_1, x_2, \ldots, x_n$ with $f(x_i) = 1/2$ and $\mu(x_i) = 1/2 + \eps_i$ for every $i \in [n]$.
        \item Elements $x'_1, x'_2, \ldots, x'_n$ with $f(x'_i) = 1/2$ and $\mu(x'_i) = 1/2 - \delta$ for every $i \in [n]$.
        \item Elements $x''_1, x''_2, \ldots, x''_n$ with $f(x''_i) = 1/2 + \eps$ and $\mu(x''_i) = 1/2$ for every $i \in [n]$.
    \end{itemize}
\end{definition}

Similar to the noiseless case, the easier direction is to show that every ``Yes'' instance of Balanced SSP leads to a small distance from calibration.

\begin{lemma}\label{lem:uniform-yes-instance}
    For every ``Yes'' instance of Balanced SSP, the instance $(\calD, f)$ from \Cref{def:reduction-uniform} satisfies $\CalDist_{\calD}(f) \le \frac{1}{36k}$.
\end{lemma}

\begin{proof}
    Let $\calI$ be a size-$k$ set $\calI \subseteq [n]$ with $\sum_{i \in \calI}a_i = S / 2$. Consider the following predictor:
    \[
        g(x) = \begin{cases}
            1/2, & x \in \{x_i: i \in \calI\} \cup \{x'_1, x'_2, \ldots, x'_n\},\\
            1/2 + \eps, & x \in \{x_i: i \in \overline{\calI}\} \cup \{x''_1, x''_2, \ldots, x''_n\}.
        \end{cases}
    \]
    We first note that $g \in \Cal(\calD)$. Conditioning on predicting $1/2$, the probability of label $1$ is
    \[
        \frac{\sum_{i \in \calI}\mu(x_i) + \sum_{i=1}^{n}\mu(x'_i)}{n + k}
    =   \frac{\sum_{i \in \calI}(1/2 + \eps_i) + (1/2 - \delta) \cdot n}{3k}
    =   \frac{1}{2} + \frac{\sum_{i \in \calI}(a_i / S) - n\delta}{3k}
    =   \frac{1}{2},
    \]
    where the last step follows from $\sum_{i \in \calI}a_i = S/2$ and $n\delta = 2k \cdot \frac{1}{4k} = 1/2$. The conditional probability of $1$ given prediction $1/2 + \eps$ is
    \[
        \frac{\sum_{i \in \overline{\calI}}\mu(x_i) + \sum_{i=1}^{n}\mu(x''_i)}{n + k}
    =   \frac{\sum_{i \in \overline{\calI}}(1/2 + \eps_i) + (1/2) \cdot n}{3k}
    =   \frac{1}{2} + \frac{\sum_{i \in \overline{\calI}}(a_i / S)}{3k}
    =   \frac{1}{2} + \eps.
    \]
    Finally, we have $d_{\calD}(f, g) = \frac{|\overline{\calI}|}{|\calX|} \cdot \eps = \frac{k}{6k} \cdot \frac{1}{6k} = \frac{1}{36k}$, since we only changed the prediction on the $x$-elements indexed by $\overline{\calI}$---each with a marginal probability of $\frac{1}{|\calX|}$---from $1/2$ to $1/2 + \eps$.
\end{proof}

In the remainder of this section, we prove that every ``No'' instance of Balanced SSP leads to $\CalDist_{\calD}(f) > \frac{1}{36k}$. Recall from \Cref{fact:CalDist-vs-partition} that each perfectly calibrated predictor corresponds to a partition of $\calX = \{x_i: i \in [n]\} \cup \{x'_i: i \in [n]\} \cup \{x''_i: i \in [n]\}$. Our proof consists of the following three steps:
\begin{itemize}
    \item We define a class of \emph{regular} partitions. Roughly speaking, a regular partition pairs exactly $k$ $x$-elements with the $2k$ $x'$-elements, and pairs the remaining $k$ $x$-elements with the $2k$ $x''$-elements. Note that the predictor in the proof of \Cref{lem:uniform-yes-instance} corresponds to a regular partition.
    \item We prove that every partition that is not regular has a cost of $\ge \frac{1}{36k} + \Omega(1/k^2)$. This is done by considering a \emph{rounded} instance in which every $\eps_i$ is replaced by $\frac{1}{2k}$, so that there are only three different types of elements (namely, $x$-, $x'$-, and $x''$-elements). Since each $\eps_i$ is $\frac{1}{2k} \pm O(1/k^2)$, this rounding only results in an $O(1/k^2)$ additive change to the cost (by \Cref{lem:partition-cost-continuity-in-D}), which will be dominated by the sub-optimality gap of every irregular partition.
    \item Finally, we argue that a regular partition gives a cost of $1/(36k)$ only if the Balanced SSP instance admits a feasible solution.
\end{itemize}

We first formally define the notion of regular partitions of $\calX$.
\begin{definition}[Regular partition]
    The signature of $\calX' \subseteq \calX$ is a triple $(a, b, c) \in \{0, 1, \ldots, n\}^3$, where $a$, $b$, and $c$ denote the numbers of $x$-, $x'$-, and $x''$-elements in $\calX'$, respectively. A partition of $\calX$ is regular if every part in it has a signature of either $(t, 2t, 0)$ or $(t, 0, 2t)$ for some integer $t$ that may depend on the part.
\end{definition}

Note that for ``Yes'' instances, the partition that witnesses $\CalDist_{\calD}(f) \le \frac{1}{36k}$ in \Cref{lem:uniform-yes-instance} has two parts of signatures $(k, 2k, 0)$ and $(k, 0, 2k)$ respectively, and is thus regular.

\subsection{Irregular Partitions are Inefficient}
We show that the cost of every irregular partition of $\calX$ is strictly higher than $\frac{1}{36k}$. To this end, we consider the following \emph{rounded} version of the instance from \Cref{def:reduction-uniform}. We will refer to the instance  from \Cref{def:reduction-uniform} as the \emph{actual} instance to avoid confusion. 

\begin{definition}[Rounded instance]\label{def:rounded-instance}
    The rounded instance for a Balanced SSP instance $(a_1, a_2, \ldots, a_n)$ is identical to the actual instance in \Cref{def:reduction-uniform}, except that $\mu(x_i) = \frac{1}{2} + \frac{1}{2k}$ for every $i \in [n]$.
\end{definition}

Recall the cost of a partition from \Cref{def:partition-cost}. We note that the distribution in the rounded instance is very close to the one in the actual instance in the TV distance. Thus, by \Cref{lem:partition-cost-continuity-in-D}, the cost of each partition only changes by $O(1/k^2)$.

\begin{lemma}\label{lem:rounding-error}
    Let $\calD$ and $\calD'$ denote the distributions in the actual instance and the rounded instance. Let $f$ be the (shared) predictor in both instances. Then, for every partition $\calP$ of $\calX$, we have
    \[
        |\cost_{\calD, f}(\calP) - \cost_{\calD', f}(\calP)| \le \frac{1}{120k^2} < \frac{1}{36k^2}.
    \]
\end{lemma}

\begin{proof}[Proof of \Cref{lem:rounding-error}]
    It suffices to prove $\dTV(\calD, \calD') \le \frac{1}{600k^2}$, which, together with \Cref{lem:partition-cost-continuity-in-D}, implies
    \[
        |\cost_{\calD, f}(\calP) - \cost_{\calD', f}(\calP)| \le 5 \cdot \dTV(\calD, \calD') \le \frac{1}{120k^2}.
    \]

    Let $\gamma$ be a shorthand for $\frac{1}{100k}$. Recall from \Cref{def:balanced-SSP} that the Balanced SSP instance satisfies $\frac{\max_{i \in [n]} a_i}{\min_{i \in [n]} a_i} \le 1 + \gamma$. Thus, for every $i \in [n]$,
    \[
        n \cdot \frac{a_i}{1 + \gamma} \le n \cdot \min_{j \in [n]}a_j \le S \le n \cdot \max_{j \in [n]}a_j \le n \cdot (1 + \gamma) a_i.
    \]
    It follows that $\frac{1}{(1 + \gamma) n} \le \frac{a_i}{S} \le \frac{1 + \gamma}{n}$, which further implies
    \[
        \left|\frac{a_i}{S} - \frac{1}{2k}\right|
    =   \left|\frac{a_i}{S} - \frac{1}{n}\right|
    \le \frac{\gamma}{n}
    =   \frac{1}{200k^2}.
    \]

    Recall from \Cref{def:reduction-uniform,def:rounded-instance} that $\calD$ and $\calD'$ have the same $\calX$-marginal and put the same probability mass on every point of form $(x'_i, y)$ and $(x''_i, y)$ (where $i \in [n]$ and $y \in \zo$). Therefore, we have the desired bound
    \begin{align*}
        \dTV(\calD, \calD')
    =   \sum_{i=1}^{n}\calD_x(x_i) \cdot \left|\mu_{\calD}(x_i) - \mu_{\calD'}(x_i)\right|
    &=  \frac{1}{6k}\sum_{i=1}^{n}\left|\left(\frac{1}{2} + \frac{a_i}{S}\right) - \left(\frac{1}{2} + \frac{1}{2k}\right)\right|\\
    &\le\frac{1}{6k} \cdot 2k \cdot \frac{1}{200k^2}
    =   \frac{1}{600k^2}.
    \end{align*}
\end{proof}

In the rounded instance, all the $n$ $x$-elements are indistinguishable. The same is true for the $x'$- and $x''$-elements, respectively. Therefore, the cost of a subset of $\calX$ is determined by its signature. The technical lemma below, which we prove in Appendix~\ref{app:hardness}, gives a characterization of all possible signatures.

\begin{lemma}\label{lem:integer-triple}
    For any integers $a, b, c \ge 0$, at least one of the following holds:
    \begin{itemize}
        \item $(a, b, c)$ is a multiple of either $(1, 2, 0)$ or $(1, 0, 2)$.
        \item $a > \frac{b + c}{2}$.
        \item For $\mu \coloneqq \frac{6a - 3b}{a + b + c}$, $(a + b) \cdot \max\{-\mu, 0\} + c \cdot \max\{2 - \mu, 0\} \ge 1$.
    \end{itemize}
\end{lemma}

In words, \Cref{lem:integer-triple} states that any subset of $\calX$ must satisfy at least one of the following: (1) its signature is allowed in regular partitions; (2) it is \emph{imbalanced} in the sense that it contains much more $x$-elements than $x'$- and $x''$-elements; (3) it is \emph{inefficient} in that it leads to a significant ``leftward movement''. Roughly speaking, the ``leftward movement'' refers to replacing predictions with smaller values (to their left on the number axis); this will be formalized in the proof below.

Now we are ready to analyze the irregular partitions.

\begin{lemma}\label{lem:irregular-partition}
    For every instance of Balanced SSP, in the actual instance from \Cref{def:reduction-uniform}, the cost of every irregular partition is strictly higher than $\frac{1}{36k}$.
\end{lemma}

\begin{proof}[Proof of \Cref{lem:irregular-partition}]
    Fix an irregular partition $\calP = \{\calX_1, \calX_2, \ldots, \calX_m\}$ of $\calX$. By \Cref{lem:rounding-error}, it suffices to prove that $\calP$ has a cost of $\ge \frac{1}{36k} + \frac{1}{36k^2}$ in the rounded instance from \Cref{def:rounded-instance}. In the remainder of this proof, we will focus on this rounded instance. In particular, all costs ($\cost(\cdot)$) and averages ($\mu(\cdot)$) will be over the distribution in \Cref{def:rounded-instance}, unless stated otherwise.

    \paragraph{Leftward and rightward costs.} Recall from \Cref{def:partition-cost} that a subset $\calX' \subseteq \calX$ has a cost of
    \[
        \cost(\calX') \coloneqq \sum_{x \in \calX'}\calD_x(x) \cdot |f(x) - \mu(\calX')|,
    \]
    and the cost of partition $\calP$ is the total cost of all parts. We further decompose the cost into two parts: the \emph{leftward cost}
    \[
        \lcost(\calX') \coloneqq \sum_{x \in \calX'}\calD_x(x) \cdot \max\{f(x) - \mu(\calX'), 0\}
    \]
    and the \emph{rightward cost}
    \[
        \rcost(\calX') \coloneqq \sum_{x \in \calX'}\calD_x(x) \cdot \max\{\mu(\calX') - f(x), 0\}.
    \]
    The leftward (resp., rightward) cost accounts for replacing the prediction $f(x)$ with a smaller (resp., larger) value to its left (resp., right). Note that $\cost(\calX') = \lcost(\calX') + \rcost(\calX')$, which implies
    \begin{equation}\label{eq:cost-sum}
        \cost(\calP) = \lcost(\calP) + \rcost(\calP).
    \end{equation}
    
    Next, we note that the difference $\rcost - \lcost$ also has a nice form. By definition,
    \begin{align*}
        \rcost(\calX') - \lcost(\calX')
    &=  \sum_{x \in \calX'}\calD_x(x) \cdot \left[\mu(\calX') - f(x)\right]\\
    &=  \calD_x(\calX') \cdot \mu(\calX') - \sum_{x \in \calX'}\calD_x(x) \cdot f(x)\\
    &=  \sum_{x \in \calX'}\calD_x(x) \cdot \mu(x) - \sum_{x \in \calX'}\calD_x(x) \cdot f(x).
    \end{align*}
    Summing the above over all parts in $\calP$ gives
    \[
        \rcost(\calP) - \lcost(\calP)
    =   \sum_{x \in \calX}\calD_x(x) \cdot \mu(x) - \sum_{x \in \calX}\calD_x(x) \cdot f(x)
    =   \Ex{(x, y) \sim \calD}{y} - \Ex{x \sim \calD_x}{f(x)}.
    \]
    The right-hand side above can be further simplified as follows. In the rounded instance, we have
    \begin{align*}
        \Ex{(x, y) \sim \calD}{y}
    &=   \frac{1}{3}\cdot\left[\left(1/2 + \frac{1}{2k}\right) + \left(1/2 - \frac{1}{4k}\right) + \frac{1}{2}\right]
    =   \frac{1}{2} + \frac{1}{12k},\\
        \text{and}\quad
        \Ex{x \sim \calD_x}{f(x)}
    &=   \frac{2}{3}\cdot\frac{1}{2} + \frac{1}{3}\cdot\left(\frac{1}{2} + \frac{1}{6k}\right)
    =   \frac{1}{2} + \frac{1}{18k}.
    \end{align*}
    Therefore, we have the identity
    \begin{equation}\label{eq:cost-diff}
        \rcost(\calP) - \lcost(\calP) = \frac{1}{36k}.
    \end{equation}

    Combining \Cref{eq:cost-sum,eq:cost-diff} gives
    \begin{equation}\label{eq:cost-vs-leftward-cost}
        \cost(\calP) = \frac{1}{36k} + 2\lcost(\calP).
    \end{equation}
    Thus, it remains to argue that partition $\calP$ contains a part with a high leftward cost.

    \paragraph{Find a candidate part.}
    For each $i \in [m]$, let $(a_i, b_i, c_i)$ be the signature of $\calX_i$. Let $\calI$ denote the set of indices $i \in [m]$ such that $(a_i, b_i, c_i)$ is proportional to neither $(1, 2, 0)$ nor $(1, 0, 2)$. Since $\calP$ is irregular, $\calI$ is non-empty. Note that $\sum_{i=1}^{m}(a_i, b_i, c_i) = (2k, 2k, 2k)$, which implies
    \[
        \sum_{i=1}^{m}\left(a_i - \frac{b_i + c_i}{2}\right)
    =   0.
    \]
    Since $a_i - \frac{b_i + c_i}{2} = 0$ holds for every $i \in [m] \setminus \calI$, we also have
    \[
        \sum_{i \in \calI}\left(a_i - \frac{b_i + c_i}{2}\right)
    =   0.
    \]
    By an averaging argument, there exists $\istar \in \calI$ such that $a_{\istar} \le \frac{b_{\istar} + c_{\istar}}{2}$. Then, $(a_{\istar}, b_{\istar}, c_{\istar})$ satisfies neither of the first two conditions in \Cref{lem:integer-triple}, so the last condition must hold, i.e.,
    \begin{equation}\label{eq:high-leftward-cost}
        (a_{\istar} + b_{\istar}) \cdot \max\left\{-\frac{6a_{\istar} - 3b_{\istar}}{a_{\istar} + b_{\istar} + c_{\istar}}, 0\right\} + c_{\istar} \cdot \max\left\{2 - \frac{6a_{\istar} - 3b_{\istar}}{a_{\istar} + b_{\istar} + c_{\istar}}, 0\right\} \ge 1.
    \end{equation}

    \paragraph{Lower bound the leftward cost.} It remains to lower bound $\lcost(\calX_{\istar})$. For brevity, we drop the subscript $\istar$ in $(a_{\istar}, b_{\istar}, c_{\istar})$ in the following. Note that
    \[
        \mu(\calX_{\istar})
    =   \frac{a \cdot [1/2 + 1/(2k)] + b \cdot [1/2 - 1/(4k)] + c \cdot (1/2)}{a + b + c}
    =   \frac{1}{2} + \frac{1}{12k} \cdot \frac{6a - 3b}{a + b + c}.
    \]
    The leftward cost of $\calX_{\istar}$ is then given by
    \begin{align*}
        &~\frac{a + b}{6k} \cdot \max\{1/2 - \mu(\calX_{\istar}), 0\} + \frac{c}{6k} \cdot \max\left\{1/2 + \frac{1}{6k} - \mu(\calX_{\istar}), 0\right\}\\
    =   &~\frac{1}{72k^2} \cdot \left[(a + b) \cdot \max\left\{-\frac{6a - 3b}{a + b + c}, 0\right\} + c \cdot \max\left\{2 -\frac{6a - 3b}{a + b + c}, 0\right\}\right]
    \ge \frac{1}{72k^2},
    \end{align*}
    where the last step follows from \Cref{eq:high-leftward-cost}.

    This shows that every irregular partition $\calP$ must give $\lcost(\calP) \ge \frac{1}{72k^2}$. By \Cref{eq:cost-vs-leftward-cost}, we have $\cost(\calP) \ge \frac{1}{36k} + \frac{1}{36k^2}$ over the rounded instance from \Cref{def:rounded-instance}. Finally, by \Cref{lem:rounding-error}, the same partition must have a cost that is strictly higher than $\frac{1}{36k}$ in the actual instance from \Cref{def:reduction-uniform}.
\end{proof}

\subsection{Regular Partitions are Also Inefficient}
In light of \Cref{lem:irregular-partition}, the only remaining piece is a lower bound on the costs of regular partitions. 

\begin{lemma}\label{lem:regular-partition}
    For every ``No'' instance of Balanced SSP, in the actual instance from \Cref{def:reduction-uniform}, the cost of every regular partition is strictly higher than $\frac{1}{36k}$.
\end{lemma}

\begin{proof}[Proof of \Cref{lem:regular-partition}]
    Fix a regular partition of $\calX$. By definition, the signature of each part is of form either $(t, 2t, 0)$ or $(t, 0, 2t)$. We will first show that we can merge multiple parts without increasing the total cost, so it is without loss of generality to assume that there are only two parts. Then, we show that the resulting partition must have a cost strictly higher than $\frac{1}{36k}$, as the only way to achieve a cost of $\frac{1}{36k}$ is to find a feasible solution to the Balanced SSP instance.

    \paragraph{Merging parts of signature $(t, 2t, 0)$.} Consider a part that contains $t$ $x$-elements $\{x_i: i \in \calI\}$ for a size-$t$ set  $\calI \subseteq [n]$ as well as $2t$ $x'$-elements. Recall that $\mu(x_i) = 1/2 + \eps_i = 1/2 + a_i / S$ and $\mu(x'_i) = 1/2 - \delta = 1/2 - 1/(4k)$. The new prediction for these elements will be
    \[
        \mu \coloneqq \frac{\sum_{i \in \calI}(1/2 + \eps_i) + 2t \cdot (1/2 - \delta)}{3t}
    =   \frac{1}{2} + \frac{\frac{\sum_{i \in \calI}a_i}{S} - \frac{t}{2k}}{3t}
    =   \frac{1}{2} + \frac{\sum_{i \in \calI}\left(\frac{a_i}{S} - \frac{1}{2k}\right)}{3t}.
    \]
    The resulting cost---for moving $3t$ elements from $1/2$ to $\mu$---would then be
    \[
        \frac{3t}{6k} \cdot |\mu - 1/2|
    =   \frac{1}{6k}\left|\sum_{i \in \calI}\left(\frac{a_i}{S} - \frac{1}{2k}\right)\right|.
    \]

    Now, suppose that the partition contains $m$ parts of signatures $(t_1, 2t_1, 0), (t_2, 2t_2, 0), \ldots, (t_m, 2t_m, 0)$ respectively. We will show that merging them into one part does not increase the cost. Let $\calI_i$ be the set of indices of the $x$-elements in the part of signature $(t_i, 2t_i, 0)$. Then, the total cost for the union of the $m$ parts would be
    \[
        \frac{1}{6k}\left|\sum_{i \in \bigcup_{j=1}^{m}\calI_j}\left(\frac{a_i}{S} - \frac{1}{2k}\right)\right|
    =   \frac{1}{6k}\left|\sum_{j=1}^{m}\sum_{i \in \calI_j}\left(\frac{a_i}{S} - \frac{1}{2k}\right)\right|
    \le \frac{1}{6k}\sum_{j=1}^{m}\left|\sum_{i \in \calI_j}\left(\frac{a_i}{S} - \frac{1}{2k}\right)\right|,
    \]
    where the second step applies the triangle inequality. Note that the right-hand side above is exactly the sum of the costs of the $m$ individual parts.

    \paragraph{Merging parts of signature $(t, 0, 2t)$.} Next, consider a part that contains $t$ $x$-elements $\{x_i: i \in \calI\}$ for a size-$t$ set  $\calI \subseteq [n]$ as well as $2t$ $x''$-elements. Again, recall that $\mu(x_i) = 1/2 + \eps_i = 1/2 + a_i / S$, $f(x_i) = 1/2$, $\mu(x''_i) = 1/2$, and $f(x''_i) = 1/2 + \eps = 1/2 + \frac{1}{6k}$. The new prediction for these elements will be
    \[
        \mu \coloneqq \frac{\sum_{i \in \calI}(1/2 + \eps_i) + 2t \cdot (1/2)}{3t}
    =   \frac{1}{2} + \frac{\frac{\sum_{i \in \calI}a_i}{S}}{3t}.
    \]
    Here, we need to move $t$ elements from $1/2$ to $\mu$ and $2t$ elements from $1/2 + \eps$ to $\mu$. Thus, the total cost is
    \begin{align*}
        \frac{t}{6k} \cdot |\mu - 1/2| + \frac{2t}{6k} \cdot |\mu - (1/2 + \eps)|
    &=  \frac{1}{6k} \cdot \left[\frac{1}{3} \cdot \frac{\sum_{i \in \calI}a_i}{S} + 2t \cdot \left|\frac{\frac{\sum_{i \in \calI}a_i}{S}}{3t} - \frac{1}{6k}\right|\right]\\
    &=  \frac{1}{6k} \cdot \left[\frac{1}{3} \cdot \frac{\sum_{i \in \calI}a_i}{S} + \frac{2}{3} \cdot \left|\sum_{i \in \calI}\left(\frac{a_i}{S} - \frac{1}{2k}\right)\right|\right].
    \end{align*}

    When merging different parts of signatures $(t_1, 0, 2t_1), (t_2, 0, 2t_2), \ldots$, the first term above is additive, while the second term is sub-additive. Therefore, we may also merge all parts of this form into a single part without increasing the total cost.

    \paragraph{Put everything together.} So far, we have reduced the number of parts to $2$: there is one part of signature $(k, 2k, 0)$ and another of signature $(k, 0, 2k)$. Let $\calI \subseteq [n]$ be the size-$k$ set that contains the indices of the $x$-elements in the part of signature $(k, 2k, 0)$. Then, the total cost of the partition is given by
    \[
        \frac{1}{6k} \cdot \left[\left|\sum_{i \in \calI}\frac{a_i}{S} - \frac{1}{2}\right| + \frac{1}{3}\sum_{i \in \overline{\calI}}\frac{a_i}{S} + \frac{2}{3}\left|\sum_{i \in \overline{\calI}}\frac{a_i}{S} - \frac{1}{2}\right|\right].
    \]

    Let $\alpha \coloneqq \sum_{i \in \calI}\frac{a_i}{S}$. We have $\sum_{i \in \overline{\calI}}\frac{a_i}{S} = 1 - \alpha$. The total cost is then simplified into
    \[
        \frac{1}{6k} \cdot \left[|\alpha - 1/2| + \frac{1 - \alpha}{3} + \frac{2}{3} \cdot |(1 - \alpha) - 1/2|\right]
    =   \frac{1}{6k} \cdot \left[\frac{1 - \alpha}{3} + \frac{5}{3} \cdot |\alpha - 1/2|\right].
    \]

    As a function of $\alpha$, the above is uniquely minimized at $\alpha = 1/2$, with a value of $\frac{1}{36k}$. Therefore, for the total cost to be lower than or equal to $\frac{1}{36k}$, there must exist a size-$k$ set $\calI \subseteq [n]$ such that $\sum_{i \in \calI}a_i = \frac{S}{2}$, which contradicts the assumption that the Balanced SSP instance is a ``No'' instance.
\end{proof}

Now, we finish the proof of \Cref{thm:NP-hard}.
\begin{proof}[Proof of \Cref{thm:NP-hard}, the second part]
    Given a Balanced SSP instance, we construct $(\calD, f)$ according to \Cref{def:reduction-uniform} in polynomial time. By \Cref{lem:uniform-yes-instance,lem:irregular-partition,lem:regular-partition}, $\CalDist_{\calD}(f) \le \frac{1}{36k^2}$ if and only if the correct answer for the Balanced SSP instance is ``Yes''. Therefore, the Balanced SSP---which is $\NP$-hard by \Cref{lem:balanced-SSP-NP-hard}---is polynomial-time reducible to computing the distance from calibration in the uniform case.
\end{proof}

%% file: sample_UB.tex
\section{Convergence of the Empirical Distance from Calibration}\label{sec:sample-UB}

In this section, we focus on the \emph{empirical} distance from calibration over the empirical distribution $\Dhat$ over the sample. We first note that \Cref{lem:continuity-in-D} implies a sample upper bound for $\CalDist_{\Dhat}(f)$ to converge to $\CalDist_{\calD}(f)$.

\begin{proof}[Proof of \Cref{thm:sample-UB}, the first part]
    It is well-known (e.g., \cite[Theorem 1.3]{Canonne22}) that a sample size of $m = O(|\calX| / \eps^2)$ is sufficient for $\dTV(\Dhat, \calD) \le \eps / 5$ to hold with probability $0.99$. By \Cref{lem:continuity-in-D}, this implies that $|\CalDist_{\Dhat}(f) - \CalDist_{\calD}(f)| \le 5 \cdot \dTV(\Dhat, \calD) = \eps$ holds with probability $0.99$.
\end{proof}

To prove the second part of \Cref{thm:sample-UB}, we first apply \Cref{lem:prediction-sparsification} to obtain a prediction-sparse predictor $g \in \Cal(\calD)$ with a size-$O(1/\eps)$ range such that $d_{\calD}(f, g) \le \CalDist_{\calD}(f) + O(\eps)$. Then, by a concentration argument, $d_{\Dhat}(f, g)$ is also small. Finally, we control the bias of $g$ (over the empirical distribution $\Dhat$) conditioning on predicting each of the $O(1/\eps)$ values.

\begin{proof}[Proof of \Cref{thm:sample-UB}, the second part]
    Let $k \ge 1$ be an integer to be chosen later. By \Cref{lem:prediction-sparsification}, there exists $g \in \Cal(\calD)$ that predicts at most $k$ different values such that $d_{\calD}(f, g) \le \CalDist_{\calD}(f) + 1/k$.


    \paragraph{$g$ is close to $f$ over $\Dhat$.} We first upper bound $d_{\Dhat}(f, g)$. Let $(x_1, y_1), (x_2, y_2), \ldots, (x_m, y_m) \in \calX \times \zo$ denote the $m$ samples from $\calD$. Note that
    \[
        d_{\Dhat}(f, g) = \frac{1}{m}\sum_{i=1}^{m}\left|f(x_i) - g(x_i)\right|
    \]
    is the average of $m$ independent random variables with the same mean $d_{\calD}(f, g) = \Ex{x \sim \calD_x}{|f(x) - g(x)|}$. By a Chernoff bound, it holds with probability at least $1 - 1/200$ that
    \[
        d_{\Dhat}(f, g) \le d_{\calD}(f, g) + O(1/\sqrt{m})
    \le \CalDist_{\calD}(f) + 1/k + O(1/\sqrt{m}).
    \]

    \paragraph{$g$ is almost calibrated over $\Dhat$.} Let $n_i \coloneqq \sum_{j=1}^{m}\1{x_j \in \calX_i}$ and $p_i \coloneqq \calD_x(\calX_i)$ for each $i \in [k]$. Over the randomness in the size-$m$ sample, $n_i$ follows $\Binomial(m, p_i)$. To upper bound $\CalDist_{\Dhat}(g)$, we instead bound the quantity
    \begin{align*}
        \cost_{\Dhat, g}(\{\calX_1, \calX_2, \ldots, \calX_k\})
    &=  \sum_{i=1}^{k}\frac{n_i}{m} \cdot \left|\mu(\calX_i) - \frac{\sum_{j=1}^{m}y_j \cdot \1{x_j \in \calX_i}}{n_i}\right|,\\
    &=  \frac{1}{m}\sum_{i=1}^{k}\left|\sum_{j=1}^{m}y_j \cdot \1{x_j \in \calX_i} - n_i \cdot \mu(\calX_i)\right|,
    \end{align*}
    which is an upper bound on $\CalDist_{\Dhat}(g)$.

    Conditioned on the values of $n_1, n_2, \ldots, n_k$, each summation $\sum_{j=1}^{m}y_j \cdot \1{x_j \in \calX_i}$ follows the distribution $\Binomial(n_i, \mu(\calX_i))$. By Jensen's inequality, the conditional expectation of
    \[
        \left|\sum_{j=1}^{m}y_j \cdot \1{x_j \in \calX_i} - n_i \cdot \mu(\calX_i)\right|
    \]
    is at most the standard deviation $\sqrt{n_i \cdot \mu(\calX_i) \cdot [1 - \mu(\calX_i)]}
    \le \frac{1}{2}\sqrt{n_i}$.

    Now we account for the randomness in $n_1, n_2, \ldots, n_k$. Again, by Jensen's inequality, we have $\Ex{}{\sqrt{n_i}}
    \le \sqrt{\Ex{}{n_i}}
    =   \sqrt{m \cdot p_i}$. Therefore,
    \begin{align*}
        \Ex{}{\CalDist_{\Dhat}(g)}
    &\le\Ex{}{\cost_{\Dhat, g}(\{\calX_1, \calX_2, \ldots, \calX_k\})}\\
    &\le\frac{1}{2m}\sum_{i=1}^{k}\Ex{}{\sqrt{n_i}}\\
    &\le\frac{1}{2\sqrt{m}}\sum_{i=1}^{k}\sqrt{p_i}
    \le \frac{\sqrt{k}}{2\sqrt{m}}.
    \end{align*}
    Then, by Markov's inequality, it holds with probability at least $1 - 1/200$ that $\CalDist_{\Dhat}(g) \le \frac{100\sqrt{k}}{\sqrt{m}}$.

    \paragraph{Put everything together.} So far, we have shown that $d_{\Dhat}(f, g) \le \CalDist_{\calD}(f) + 1/k + O(1/\sqrt{m})$ and $\CalDist_{\Dhat}(g) \le \frac{100\sqrt{k}}{\sqrt{m}}$ each holds with probability $\ge 1 - 1/200$. Set $k = \Theta(m^{1/3})$. By the union bound and \Cref{fact:continuity-in-f}, it holds with probability $\ge 0.99$ that
    \[
        \CalDist_{\Dhat}(f)
    \le \CalDist_{\Dhat}(g) + d_{\Dhat}(f, g)
    \le \CalDist_{\calD}(f) + O(m^{-1/3}).
    \]
    For some $m = O(1/\eps^3)$, the right-hand side above is at most $\CalDist_{\calD}(f) + \eps$.
\end{proof}

%% file: sample_LB.tex
\section{Sample Lower Bound for Estimation}\label{sec:sample-LB}
We prove the first part of \Cref{thm:sample-LB} in this section. We first define the hard instance in \Cref{sec:sample-LB-hard-instance}, which consists of a ``pure'' variant and a ``mixed'' variant. We then show that every accurate estimator for the calibration distance distinguishes the two variants. Finally, we prove that distinguishing requires $\Omega(\sqrt{|\calX|})$ samples via a standard birthday-paradox argument.

\subsection{The Hard Instance}\label{sec:sample-LB-hard-instance}
Let $k$ be a positive integer and $\gamma \in (0, 1)$. Consider a size-$(4k + 3)$ domain
\[
    \calX \coloneqq \{\bot, x^{-}, x^{+}, x_1, x_2, \ldots, x_{4k}\},
\]
where:
\begin{itemize}
    \item $f(\bot) = 1/2$, $f(x^{-}) = 1/3$, $f(x^{+}) = 2/3$, and $f(x_i) = 1/2$ for every $i \in [4k]$.
    \item $\calD_x(\bot) = 1 - \gamma$, $\calD_x(x^{-}) = \calD_x(x^{+}) = \gamma/4$ and $\calD_x(x_i) = \frac{\gamma}{8k}$ for every $i \in [4k]$.
    \item $\mu(\bot) = \mu(x^{-}) = \mu(x^{+}) = 1/2$.
\end{itemize}

It remains to specify the value of $\mu(x_i)$ for each $i \in [4k]$. We consider the following two cases:

\begin{itemize}
    \item \textbf{The pure case:} $\mu(x_i) = 1/2$ for each $i \in [4k]$.

    \item \textbf{The mixed case:} The values $\mu(x_1), \mu(x_2), \ldots, \mu(x_{4k})$ are sampled from $\Bern(1/2)$ independently.
\end{itemize}

\subsection{Accurate Estimators Distinguish the Two Variants}

For brevity, we say that an estimator is $(\eps, \delta)$-PAC if it outputs the calibration distance up to error $\eps$ with probability at least $1 - \delta$.

\begin{lemma}\label{lem:estimation-implies-distinguishing}
    The following holds for every $(\eps, \delta)$-PAC estimator $\calA$:
    \begin{itemize}
        \item In the pure case, $\calA$ outputs a value $\ge \frac{\gamma}{12} - \eps$ with probability $\ge 1 - \delta$.
        \item In the mixed case, $\calA$ outputs a value $\le \frac{\gamma}{16} + \eps$ with probability $\ge 1 - \delta - e^{-\Omega(k)}$.
    \end{itemize}
\end{lemma}

\begin{proof}[Proof of \Cref{lem:estimation-implies-distinguishing}]
    It suffices to prove that $\CalDist_{\calD}(f) \ge \gamma / 12$ holds in the pure case, while $\CalDist_{\calD}(f) \le \gamma / 16$ holds with probability $1 - e^{-\Omega(k)}$ in the mixed case. The lemma would then follow from the assumption on $\calA$ and the union bound.

    \paragraph{The pure case.} Since $\mu(x) = 1/2$ holds for every $x \in \calX$ in the pure case, there is only one perfectly calibrated predictor, namely, the constant function $g(x) = 1/2$. Therefore,
    \[
        \CalDist_{\calD}(f)
    =   d_{\calD}(f, g)
    =   \calD_x(x^{-}) \cdot |1/3 - 1/2| + \calD_x(x^{+}) \cdot |2/3 - 1/2|
    =   \frac{\gamma}{12}.
    \]

    \paragraph{The mixed case.} Let random variable $N$ denote the number of indices $i \in [4k]$ with $\mu(x_i) = 1$. Note that $N$ follows $\Binomial(4k, 1/2)$. Let $\alpha \coloneqq 1/6$. By a Chernoff bound, it holds with probability $1 - e^{-\Omega(k)}$ that $|N - 2k| \le \alpha k$. It remains to show that this high-probability event implies $\CalDist_{\calD}(f) \le \frac{\gamma}{16}$.

    Since $|N - 2k| \le \alpha k$ implies $N \in [k, 3k]$, we can find size-$k$ subsets $\calI_0, \calI_1 \subseteq [4k]$ such that $\mu(x_i) = b$ holds for each $b \in \zo$ and $i \in \calI_b$. Let $\overline{\calI} \coloneqq [4k] \setminus (\calI_0 \cup \calI_1)$. Consider the following alternative predictor:
    \[
        g(x) = \begin{cases}
            1/3, & x \in \{x^{-}\} \cup \{x_i: i \in \calI_0\},\\
            2/3, & x \in \{x^{+}\} \cup \{x_i: i \in \calI_1\},\\
            \frac{N - k}{2k}, & x \in \{x_i: i \in \overline{\calI}\},\\
            1/2, & x = \bot.
        \end{cases}
    \]

    We first verify that $g \in \Cal(\calD)$. For brevity, we will assume that $\frac{N - k}{2k} \notin \{1/3, 1/2, 2/3\}$; the general case can be handled in a similar way. We have
    \[
        \Ex{(x, y) \sim \calD}{y \mid g(x) = 1/3}
    =   \frac{\calD_x(x^{-}) \cdot \mu(x^{-}) + \sum_{i \in \calI_0}\calD_x(x_i) \cdot \mu(x_i)}{\calD_x(x^{-}) + \sum_{i \in \calI_0}\calD_x(x_i)}
    =   \frac{\frac{\gamma}{4} \cdot \frac{1}{2} + k \cdot \frac{\gamma}{8k} \cdot 0}{\frac{\gamma}{4} + k \cdot \frac{\gamma}{8k}}
    =   \frac{1}{3}.
    \]
    The condition $\Ex{(x, y) \sim \calD}{y \mid g(x) = 2/3} = 2/3$ holds analogously. By definition of $N$, $\calI_0$ and $\calI_1$, $\overline{\calI}$ contains $N - k$ indices with $\mu(x_i) = 1$ and $3k - N$ indices with $\mu(x_i) = 0$. Then, for prediction $\frac{N - k}{2k}$, we have
    \[
        \Ex{(x, y) \sim \calD}{y \mid g(x) = \frac{N - k}{2k}}
    =   \frac{\sum_{i \in \overline{\calI}}\calD_x(x_i) \cdot \mu(x_i)}{\sum_{i \in \overline{\calI}}\calD_x(x_i)}
    =   \frac{(N - k) \cdot \frac{\gamma}{8k} \cdot 1 + (3k - N) \cdot \frac{\gamma}{8k} \cdot 0}{2k \cdot \frac{\gamma}{8k}}
    =   \frac{N - k}{2k}.
    \]
    Since $g$ only predicts $1/2$ on $\bot$ and $\mu(\bot) = 1/2$, we have $\Ex{(x, y) \sim \calD}{y \mid g(x) = 1/2} = 1/2$.

    Next, we note that $d_{\calD}(f, g)$ is given by:
    \[
        \calD_x(\{x_i: i \in \calI_0\}) \cdot |1/2 - 1/3| + \calD_x(\{x_i: i \in \calI_1\}) \cdot |1/2 - 2/3| + \calD_x(\{x_i: i \in \overline{\calI}\}) \cdot \left|1/2 - \frac{N - k}{2k}\right|.
    \]
    Both of the first two terms above are given by $k \cdot \frac{\gamma}{8k} \cdot \frac{1}{6} = \frac{\gamma}{48}$, while the last term is
    \[
        2k \cdot \frac{\gamma}{8k} \cdot \frac{|N - 2k|}{2k}
    \le \frac{\gamma}{8k} \cdot \alpha k
    =   \frac{\gamma}{48}.
    \]
    Therefore, we conclude that
    \[
        \CalDist_{\calD}(f)
    \le d_{\calD}(f, g)
    \le 3 \cdot \frac{\gamma}{48}
    =   \frac{\gamma}{16}
    \]
    holds with probability $1 - e^{-\Omega(k)}$ in the mixed case.
\end{proof}

\subsection{Lower Bound for Distinguishing}
Next, we prove via a standard birthday-paradox argument that the pure and mixed cases are indistinguishable, unless $\Omega(\sqrt{|\calX|} / \eps)$ samples are drawn. For positive integer $m$, let $\Dpure^{\otimes m}$ and $\Dmix^{\otimes m}$ be the distributions of a size-$m$ sample from the pure and mixed cases, respectively. Note that $\Dmix^{\otimes m}$ involves two levels of randomness: the randomness in parameters $\mu(x_1), \mu(x_2), \ldots, \mu(x_{4k})$ that define distribution $\calD$ and the random drawing of $m$ independent samples from $\calD$. In other words, $\Dmix^{\otimes m}$ is the uniform mixture of $2^{4k}$ $m$-fold product distributions obtained from the $2^{4k}$ possible choices of $\calD$.

\begin{lemma}\label{lem:birthday-paradox}
    \[
        \dTV\left(\Dpure^{\otimes m}, \Dmix^{\otimes m}\right) \le \frac{m^2\gamma^2}{32k}.
    \]
\end{lemma}

\begin{proof}[Proof of \Cref{lem:birthday-paradox}]
    We will construct a coupling between $\Dpure^{\otimes m}$ and $\Dmix^{\otimes m}$. Let
    \[
        S = ((x^{(1)}, y^{(1)}), (x^{(2)}, y^{(2)}), \ldots, (x^{(m)}, y^{(m)}))
    \]
    denote the size-$m$ sample from either the pure or the mixed case.
    
    Since $\mu(x) = 1/2$ holds for every $x \in \calX$ in the pure case, $\vec{y} \coloneqq (y^{(1)}, y^{(2)}, \ldots, y^{(m)})$ is uniformly distributed over $\zo^m$ conditioned on any possible realization of $\vec{x} \coloneqq (x^{(1)}, x^{(2)}, \ldots, x^{(m)})$. On the other hand, $\vec{y} \mid \vec{x}$ is uniformly distributed in the mixed case conditioned on that each of the $4k$ elements $x_1, x_2, \ldots, x_{4k}$ appears at most once in $\vec{x}$.

    By the union bound, the above implies a coupling between $\Dpure^{\otimes m}$ and $\Dmix^{\otimes m}$ that agrees except with probability at most
    \[
        \binom{m}{2} \cdot 4k \cdot \left(\frac{\gamma}{8k}\right)^2
    \le \frac{m^2\gamma^2}{32k}.
    \]
    This implies $\dTV\left(\Dpure^{\otimes m}, \Dmix^{\otimes m}\right) \le \frac{m^2\gamma^2}{32k}$.
\end{proof}

\subsection{Put Everything Together}
Now, we state and prove the formal version of \Cref{thm:sample-LB}.
\begin{theorem}\label{thm:sample-LB-formal}
    Suppose that integer $k$ is sufficiently large and $\eps, \delta \in (0, 0.01)$. Then, every $(\eps, \delta)$-PAC algorithm that estimates the distance from calibration on a domain of size $|\calX| = 4k + 3$ must draw a sample of size at least
    \[
        \frac{\sqrt{k}}{25\eps} = \Omega\left(\frac{\sqrt{|\calX|}}{\eps}\right).
    \]
\end{theorem}

\begin{proof}[Proof of \Cref{thm:sample-LB-formal}]
    Let $\gamma \coloneqq 100\eps \in (0, 1)$ and $\theta \coloneqq \frac{\gamma / 12 + \gamma / 16}{2}$. It can be verified that $\gamma / 12 - \eps > \theta > \gamma / 16 + \eps$. Let $\calA$ be an $m$-sample $(\eps, \delta)$-PAC estimator, and consider running $\calA$ in the pure and mixed cases of the hard instance with parameters $k$ and $\gamma$. By \Cref{lem:estimation-implies-distinguishing}, we have
    \[
        \pr{S \sim \Dpure^{\otimes m}}{\calA(S) > \theta}
    \ge \pr{S \sim \Dpure^{\otimes m}}{\calA(S) \ge \gamma / 12 - \eps} \ge 1 - \delta > 0.99
    \]
    and
    \[
        \pr{S \sim \Dmix^{\otimes m}}{\calA(S) > \theta}
    \le 1 - \pr{S \sim \Dmix^{\otimes m}}{\calA(S) \le \frac{\gamma}{16} + \eps} \le \delta + e^{-\Omega(k)} < 0.01 + e^{-\Omega(k)}.
    \]
    For all sufficiently large $k$, the two inequalities together imply
    \[
        \dTV(\Dpure^{\otimes m}, \Dmix^{\otimes m})
    \ge \pr{S \sim \Dpure^{\otimes m}}{\calA(S) > \theta} - \pr{S \sim \Dmix^{\otimes m}}{\calA(S) > \theta}
    >   0.98 - e^{-\Omega(k)}
    \ge \frac{1}{2}.
    \]
    By \Cref{lem:birthday-paradox}, we must have $\frac{m^2\gamma^2}{32k} \ge \frac{1}{2}$, which implies $m \ge \frac{4\sqrt{k}}{\gamma} = \frac{\sqrt{k}}{25\eps}$.
\end{proof}

\section{Tightness of the One-Sided Convergence Rate}\label{sec:one-sided-LB}
We prove the second part of \Cref{thm:sample-LB}, which we formally state below.

\begin{theorem}\label{thm:one-sided-LB-formal}
    For every sufficiently large integer $k$, there exists an instance $(\calD, f)$ over a size-$k$ domain such that $\CalDist_{\calD}(f) = 0$ but a size-$k^3$ sample leads to
    \[
        \pr{}{\CalDist_{\Dhat}(f) \ge \Omega(1/k)} \ge \Omega(1).
    \]
\end{theorem}

By setting $k = \Theta(1/\eps)$, the theorem implies that the $1/\eps^3$ dependence in \Cref{thm:sample-UB} cannot be improved in general.

\subsection{The Hard Instance and Overview of Analysis}
We consider the domain $\calX \coloneqq [k] = \{1, 2, \ldots, k\}$, predictor $f(i) = 1/3 + i / (3k)$, and distribution $\calD$ specified by:
\begin{itemize}
    \item $\calD_x(i) = 1/k$ for every $i \in [k]$.
    \item $\mu_{\calD}(i) = 1/3 + i / (3k)$ for every $i \in [k]$.
\end{itemize}
Since $f$ agrees with $\mu_{\calD}(\cdot)$ on every $x \in \calX$, we have $\CalDist_{\calD}(f) = 0$.

It remains to lower bound $\CalDist_{\Dhat}(f)$ over a sample of size $m = k^3$, denoted by
\[
(x_1, y_1), (x_2, y_2), \ldots, (x_m, y_m) \sim \calD.
\]
For each $i \in [k]$, let
\[
    n_i \coloneqq \sum_{j=1}^{m}\1{x_j = i}
\]
be the number of times $i$ gets sampled, and note that the empirical average of the label is given by
\[
    \mu_{\Dhat}(i) = \frac{\sum_{j=1}^{m}y_j \cdot \1{x_j = i}}{n_i}.
\]
Let $\Delta_i \coloneqq \mu_{\Dhat}(i) - \mu_{\calD}(i)$ be the deviation of $\mu_{\Dhat}(i)$ from its expectation.

We consider elements that are \emph{typical} in the following sense.

\begin{definition}
    An element $i \in [k]$ is typical if
    \[
        \frac{m}{2k} \le n_i \le \frac{2m}{k}
    \quad \text{and} \quad
        |\Delta_i| \ge \frac{1}{300k}.
    \]
\end{definition}

We will first apply standard concentration and anti-concentration bounds to show that, with high probability, the majority of elements are typical.

\begin{lemma}\label{lem:many-typical-elements}
    The following holds for all sufficiently large $k$: with probability at least $0.9$, there are at least $0.9 \cdot k$ typical elements.
\end{lemma}

Then, we show that the presence of many typical elements implies a lower bound on $\CalDist_{\Dhat}(f)$.

\begin{lemma}\label{lem:typical-elements-give-high-CalDist}
    If there are at least $0.9 \cdot k$ typical elements, we have
    \[
        \CalDist_{\Dhat}(f) \ge \Omega(1/k).
    \]
\end{lemma}

\Cref{thm:one-sided-LB-formal} immediately follows from \Cref{lem:many-typical-elements,lem:typical-elements-give-high-CalDist}.

\subsection{Proof of \Cref{lem:many-typical-elements}}

We will use the following corollary of the Berry-Esseen theorem, which we prove in \Cref{app:one-sided-LB}.
\begin{lemma}\label{lem:anti-concentration}
    For every $p \in [1/3, 2/3]$, $\eps > 0$ and integer $n \ge 1$, it holds that
    \[
        \pr{X \sim \Binomial(n, p)}{|X / n - p| \le \eps}
    \le \frac{3}{\sqrt{\pi}} \cdot \eps\sqrt{n} + O\left(\frac{1}{\sqrt{n}}\right).
    \]
\end{lemma}

\begin{proof}[Proof of \Cref{lem:many-typical-elements}]
    It suffices to prove that each $i \in [k]$ is typical with probability at least $0.99$. This would imply that the expected number of \emph{atypical} elements is at most $k/100$. By Markov's inequality, the probability of having at least $k/10$ atypical elements is at most $1/10$, and the lemma follows.

    Fix $i \in [k]$. Over the randomness in the size-$m$ sample, $n_i$ follows $\Binomial(m, 1/k)$. By a Chernoff bound,
    \[
        \pr{}{n_i < \frac{m}{2k} \vee n_i > \frac{2m}{k}}
    \le 2e^{-\Omega(m/k)}
    =   e^{-\Omega(k^2)},
    \]
    where the last step applies $m = k^3$. Thus, it remains to control the conditional probability that $|\Delta_i| \ge \frac{1}{300k}$ gets violated given $n_i \in [m / (2k), 2m / k]$.


    To bound the probability of $|\Delta_i| < \frac{1}{300k}$, we apply \Cref{lem:anti-concentration} with $n = n_i \in [m / (2k), 2m / k]$ and $p = \mu_{\calD}(i) \in [1/3, 2/3]$. This gives
    \[
        \pr{}{|\Delta_i| < \frac{1}{300k}}
    \le \frac{3}{\sqrt{\pi}} \cdot \frac{1}{300k} \cdot \sqrt{\frac{2m}{k}} + O\left(\frac{1}{\sqrt{m / k}}\right)
    \le \frac{\sqrt{2}}{100\sqrt{\pi}} + O(1/k).
    \]

    By the union bound, the probability that $i$ is atypical is upper bounded by
    \[
        \frac{\sqrt{2}}{100\sqrt{\pi}} + e^{-\Omega(k^2)} + O(1/k),
    \]
    which is at most $1/100$ for all sufficiently large $k$.
\end{proof}

\subsection{Proof of \Cref{lem:typical-elements-give-high-CalDist}}

We prove \Cref{lem:typical-elements-give-high-CalDist} via a lower bound on the cost of every subset of $\calX$ that depends on the number of typical and atypical elements. 

\begin{proof}[Proof of \Cref{lem:typical-elements-give-high-CalDist}]
    It suffices to lower bound $\cost_{\Dhat, f}(\calP)$ for every partition $\calP = \{\calX_1, \calX_2, \ldots, \calX_l\}$ of $\calX = [k]$. Let $T \subseteq [k]$ denote the set of all typical elements. We will prove the following inequality for every $S \subseteq [k]$:
    \begin{equation}\label{eq:cost-typical-elements}
        \cost_{\Dhat, f}(S) \ge \Omega(1/k^2) \cdot (|S \cap T| - \1{|S \cap T| = 1 \wedge |S \cap \overline{T}| \ge 1}).
    \end{equation}

    In words, \Cref{eq:cost-typical-elements} states that each typical element in the set contributes $\Omega(1/k^2)$ to the cost. The only exception is when the set contains exactly one typical element and at least one atypical element, in which case we have a trivial lower bound of $0$.

    \paragraph{\Cref{eq:cost-typical-elements} implies the lemma.} Applying the inequality to $\calX_1, \calX_2, \ldots, \calX_l$ and taking the sum gives
    \[
        \cost_{\Dhat, f}(\calP)
    \ge \Omega(1/k^2) \cdot \left[\sum_{i=1}^{l}|\calX_i \cap T| - \sum_{i=1}^{l}\1{|\calX_i \cap T| = 1 \wedge |\calX_i \cap \overline{T}| \ge 1}\right].
    \]
    The first summation $\sum_{i=1}^{l}|\calX_i \cap T|$ clearly evaluates to $|T| \ge 0.9 \cdot k$. The second summation is upper bounded by $\sum_{i=1}^{l}|\calX_i \cap \overline{T}| = |\overline{T}| \le 0.1 \cdot k$. Therefore, we have
    \[
        \cost_{\Dhat, f}(\calP) \ge \Omega(1/k^2) \cdot (0.9 - 0.1) \cdot k = \Omega(1/k).
    \]

    \paragraph{Proof of \Cref{eq:cost-typical-elements}.} We prove the inequality via the following case analysis:
    \begin{itemize}
        \item \textbf{Case 1: $|S \cap T| = 0$.} In this case, the right-hand side of \Cref{eq:cost-typical-elements} evaluates to $0$, and the inequality holds trivially.

        \item \textbf{Case 2: $|S \cap T| = 1$ and $|S \cap \overline{T}| \ge 1$.} Again, the right-hand side reduces to $0$, so the inequality holds.

        \item \textbf{Case 3: $|S \cap T| = 1$ and $|S \cap \overline{T}| = 0$.} In this case, $S = \{i\}$ for some typical element $i$. By definition, we have
        \[
            \Dhat_x(i) = \frac{n_i}{m} \ge \Omega(1/k)
        \quad \text{and} \quad
            |f(i) - \mu_{\Dhat}(i)| = |\mu_{\calD}(i) - \mu_{\Dhat}(i)| = |\Delta_i| \ge \Omega(1/k).
        \]
        It follows that
        \[
            \cost_{\Dhat, f}(S)
        =   \Dhat_x(i) \cdot |f(i) - \mu_{\Dhat}(i)|
        \ge \Omega(1/k^2),
        \]
        which matches the right-hand side.
        
        \item \textbf{Case 4: $|S \cap T| \ge 2$.} By definition of typical elements, $\Dhat_x(i) = n_i / m \ge \Omega(1/k)$ holds for every $i \in T$. It follows that
        \[
            \cost_{\Dhat, f}(S)
        \ge \sum_{i \in S \cap T}\Dhat_x(i) \cdot |f(i) - \mu_{\Dhat}(S)|
        \ge \Omega(1/k) \cdot \sum_{i \in S \cap T}|f(i) - \mu_{\Dhat}(S)|.
        \]
        Since $f(i) = 1/3 + 1/(3k)$, the values $\{f(i): i \in S \cap T\}$ are separated by at least $1/(3k)$. Then, regardless of the value of $\mu_{\Dhat}(S)$, it holds for all but at most one element $i \in S \cap T$ that $|f(i) - \mu_{\Dhat}(S)| \ge \frac{1}{6k}$. Therefore, we have
        \[
            \cost_{\Dhat, f}(S)
        \ge \Omega(1/k) \cdot (|S \cap T| - 1) \cdot \frac{1}{6k}
        \ge \Omega(1/k^2) \cdot |S \cap T|,
        \]
        where the last step applies $|S \cap T| \ge 2$.
    \end{itemize}
    This establishes \Cref{eq:cost-typical-elements} and completes the proof.
\end{proof}

%% file: appendix.tex
\section{Continuity of the Distance from Calibration}\label{app:basic}

\begin{proof}[Proof of \Cref{lem:partition-cost-continuity-in-D}]
    By symmetry, it suffices to prove
    \[
        \cost_{\calD_1, f}(\calP) \le \cost_{\calD_2, f}(\calP) + 5 \cdot \dTV(\calD_1, \calD_2).
    \]
    Write $\calP = \{\calX_1, \calX_2, \ldots, \calX_k\}$. For each $i \in [k]$, $j \in \{1, 2\}$, and $b \in \zo$, let
    \[
        m^{(b, j)}_i \coloneqq \pr{(x, y) \sim \calD_j}{x \in \calX_i \wedge y = b} = \sum_{x \in \calX_i}\calD_j(x, b)
    \]
    be the total probability mass of $\calX_i \times \{b\}$ in distribution $\calD_j$. Let $\calD_{1,x}$ and $\calD_{2,x}$ be the $\calX$-marginals of $\calD_1$ and $\calD_2$ respectively. By definition of the total variation distance, we have
    \begin{equation}\label{eq:TV-dist-1}
        \dTV(\calD_1, \calD_2)
    =   \max_{S \subseteq \calX \times \zo}[\calD_1(S) - \calD_2(S)]
    \ge \sum_{i=1}^{k}\sum_{x \in \calX_i}\max\left\{\calD_{1,x}(x) - \calD_{2,x}(x), 0\right\}
    \end{equation}
    and
    \begin{equation}\begin{split}\label{eq:TV-dist-2}
        \dTV(\calD_1, \calD_2)
    &=  \frac{1}{2}\sum_{x \in \calX}\sum_{y \in \zo}|\calD_1(x, y) - \calD_2(x, y)|\\
    &\ge \frac{1}{4}\sum_{i=1}^{k}\left[\left|m_i^{(0,1)} - m_i^{(0,2)}\right| + 2\cdot \left|m_i^{(1,1)} - m_i^{(1,2)}\right|\right].
    \end{split}\end{equation}
    
    We will show that, for every $i \in [k]$,
    \begin{equation}\begin{split}\label{eq:continuity-in-D-key-ineq}
        \cost_{\calD_1, f}(\calX_i) - \cost_{\calD_2, f}(\calX_i)
    &\le \sum_{x \in \calX_i}\max\left\{\calD_{1,x}(x) - \calD_{2,x}(x), 0\right\}\\
    &+ \left|m_i^{(0,1)} - m_i^{(0,2)}\right| + 2\cdot \left|m_i^{(1,1)} - m_i^{(1,2)}\right|.
    \end{split}\end{equation}
    Note that, by \Cref{eq:TV-dist-1,eq:TV-dist-2}, summing the above over $i \in [k]$ gives
    \[
        \cost_{\calD_1, f}(\calP)
    \le \cost_{\calD_2, f}(\calP) + \dTV(\calD_1, \calD_2) + 4 \cdot \dTV(\calD_1, \calD_2)
    =   \cost_{\calD_2, f}(\calP) + 5 \cdot \dTV(\calD_1, \calD_2)
    \]
    and thus the lemma.
    
    \paragraph{Proof of \Cref{eq:continuity-in-D-key-ineq}.} In the following, we fix a part $\calX_i$ and omit the subscript $i$ in $m^{(b,j)}_i$ for brevity. We also introduce the shorthands
    \[
        \mu_1 \coloneqq \mu_{\calD_1}(\calX_i) = \frac{m^{(1,1)}}{m^{(0,1)} + m^{(1,1)}}
    \quad \text{and} \quad
        \mu_2 \coloneqq \mu_{\calD_2}(\calX_i) = \frac{m^{(1,2)}}{m^{(0,2)} + m^{(1,2)}}.
    \]
    Note that $\cost_{\calD_1, f}(\calX_i) = \sum_{x \in \calX_i}\calD_{1,x}(x) \cdot |f(x) - \mu_1|$ and $\cost_{\calD_2, f}(\calX_i) = \sum_{x \in \calX_i}\calD_{2,x}(x) \cdot |f(x) - \mu_2|$.

    Towards upper bounding $\cost_{\calD_1, f}(\calX_i) - \cost_{\calD_2, f}(\calX_i)$, we define
    \[
        A \coloneqq \sum_{x \in \calX_i}\calD_{1,x}(x) \cdot |f(x) - \mu_2|
    \]
    and bound $\cost_{\calD_1, f}(\calX_i) - A$ and $A - \cost_{\calD_2, f}(\calX_i)$ separately.     For the latter part, we have
    \[
        A - \cost_{\calD_2, f}(\calX_i)
    =   \sum_{x \in \calX_i}[\calD_{1,x}(x) - \calD_{2,x}(x)] \cdot |f(x) - \mu_2|
    \le \sum_{x \in \calX_i}\max\{\calD_{1,x}(x) - \calD_{2,x}(x), 0\},
    \]
    where the second step holds since $f(x), \mu_2 \in [0, 1]$ implies $|f(x) - \mu_2| \le 1$.
    
    For the former part, we have
    \[
        \cost_{\calD_1, f}(\calX_i) - A
    =   \sum_{x \in \calX_i}\calD_{1,x}(x) \cdot \left[|f(x) - \mu_1| - |f(x) - \mu_2|\right]
    \le \calD_{1,x}(\calX_i) \cdot |\mu_1 - \mu_2|.
    \]
    Since $\calD_{1,x}(\calX_i) = m^{(0,1)} + m^{(1,1)}$, we have
    \[
        \calD_{1,x}(\calX_i) \cdot |\mu_1 - \mu_2|
    =   \left(m^{(0,1)} + m^{(1,1)}\right) \cdot \left|\frac{m^{(1,1)}}{m^{(0,1)} + m^{(1,1)}} - \frac{m^{(1,2)}}{m^{(0,2)} + m^{(1,2)}}\right|,
    \]
    which can be equivalently written as
    \begin{align*}
        \left|m^{(1,1)} - m^{(1,2)} \cdot \frac{m^{(0,1)} + m^{(1,1)}}{m^{(0,2)} + m^{(1,2)}}\right|
    &\le \left|m^{(1,1)} - m^{(1,2)}\right| + \left|m^{(1,2)} - m^{(1,2)} \cdot \frac{m^{(0,1)} + m^{(1,1)}}{m^{(0,2)} + m^{(1,2)}}\right|.
    \end{align*}
    The second term above is further given by
    \[
        \frac{m^{(1,2)}}{m^{(0,2)} + m^{(1,2)}}
    \cdot \left|(m^{(0,2)} + m^{(1,2)}) - (m^{(0,1)} + m^{(1,1)})\right|
    \le \left|m^{(0,1)} - m^{(0,2)}\right| + \left|m^{(1,1)} - m^{(1,2)}\right|.
    \]
    Therefore, we have
    \[
        \cost_{\calD_1, f}(\calX_i) - A
    \le \left|m^{(0,1)} - m^{(0,2)}\right| + 2\cdot \left|m^{(1,1)} - m^{(1,2)}\right|.
    \]
    This establishes \Cref{eq:continuity-in-D-key-ineq} and thus the lemma.
\end{proof}

\section{Omitted Proof for \Cref{sec:algorithm}}\label{app:algorithm}

\begin{proof}[Proof of \Cref{lem:type-monotonicity}]
    Let $g \in \Cal(\calD)$ be a predictor such that $d_{\calD}(f, g) = \CalDist_{\calD}(f)$. We will transform $g$ into a type-monotone predictor without increasing its distance from $f$.

    Let $x_{i,j}, x_{i,j+1} \in \calX$ be two elements of the same type~$i$ such that $g(x_{i,j}) > g(x_{i,j+1})$. If no such pair exists, $g$ is already type-monotone and we are done. Otherwise, we use the shorthands $a \coloneqq x_{i,j}$ and $b \coloneqq x_{i,j+1}$, and consider the alternative predictor
    \[
        g'(x) = \begin{cases}
            g(b), & x = a,\\
            g(a), & x = b,\\
            g(x), & x \notin \{a, b\}
        \end{cases}
    \]
    obtained from $g$ by swapping the predictions on $a$ and $b$. Since $\calD_x(a) = \calD_x(b)$ and $\mu_{\calD}(a) = \mu_{\calD}(b)$, it can be easily verified that $g' \in \Cal(\calD)$.

    Moreover, the increase in the distance from $f$ is given by
    \[
        d_{\calD}(f, g') - d_{\calD}(f, g)
    =   \calD_x(a) \cdot \left[|f(a) - g'(a)| - |f(a) - g(a)|\right] + \calD_x(b) \cdot \left[|f(b) - g'(b)| - |f(b) - g(b)|\right],
    \]
    which is in turn $\calD_x(a)$ times
    \[
        |f(a) - g(b)| + |f(b) - g(a)| - |f(a) - g(a)| - |f(b) - g(b)|.
    \]
    Since $f(a) \le f(b)$ and $g(a) > g(b)$, the above is non-positive by the rearrangement inequality. It follows that
    \[
        \CalDist_{\calD}(f) \le d_{\calD}(f, g') \le d_{\calD}(f, g) = \CalDist_{\calD}(f),
    \]
    and thus $d_{\calD}(f, g') = \CalDist_{\calD}(f)$.

    Therefore, we can keep adjusting the predictor $g$ while maintaining the invariant that $g \in \Cal(\calD)$ and $d_{\calD}(f, g) = \CalDist_{\calD}(f)$. Since each adjustment strictly reduces the number of inversions, this procedure must terminate, at which point we obtain a type-monotone predictor that proves the lemma.
\end{proof}

\section{Omitted Proofs for \Cref{sec:hardness}}\label{app:hardness}
We prove that Balanced SSP (\Cref{def:balanced-SSP}) is $\NP$-hard by reducing from the partition problem, which is one of Karp's $21$ $\NP$-complete problems~\cite{Karp72}.

\begin{definition}[The partition problem]
    Given $n$ positive integers $a_1, a_2, \ldots, a_n$, decide whether there exists $\calI \subseteq [n]$ such that $\sum_{i \in \calI}a_i = \frac{1}{2}\sum_{i=1}^{n}a_i$.
\end{definition}

\begin{proof}[Proof of \Cref{lem:balanced-SSP-NP-hard}]
    Let $a = (a_1, a_2, \ldots, a_n)$ be an instance of the partition problem. Set $S \coloneqq \sum_{i=1}^{n}a_i$ and $M \coloneqq 100nS$. Consider the length-$2n$ list
    \[
        a' \coloneqq (a_1 + M, a_2 + M, \ldots, a_n + M, M, M, \ldots, M)
    \]
    obtained by adding $M$ to all entries of $a$ and padding $n$ copies of $M$. Note that $\sum_{i=1}^{2n}a'_i = S + 2nM$.

    We first verify that $a'$ is a valid instance of Balanced SSP. Note that all entries of $a'$ are between $M$ and $M + S$, so $\frac{\max_{i \in [2n]} a'_i}{\min_{i \in [2n]} a'_i} \le 1 + \frac{S}{M} = 1 + \frac{1}{100n}$. Next, we show that $a$ is a ``Yes'' instance of the partition problem if and only if $a'$ is a ``Yes'' instance of Balanced SSP. For the ``if'' direction, suppose that $\calI \subseteq [2n]$ is a size-$n$ set with $\sum_{i \in \calI}a'_i = \frac{1}{2}\sum_{i=1}^{2n}a'_i = S/2 + nM$. Then, letting $\calI' \coloneqq \calI \cap [n]$, we have
    \[
        \sum_{i \in \calI'}a_i
    =   \sum_{i \in \calI}(a'_i - M)
    =   \sum_{i \in \calI}a'_i - nM = \frac{S}{2}.
    \]
    Conversely, if $\calI \subseteq [n]$ satisfies $\sum_{i \in \calI}a_i = S/2$, it holds for $\calI' \coloneqq \calI \cup \{n + 1, n + 2, \ldots, 2n - |\calI|\}$ that $|\calI'| = n$ and
    \[
        \sum_{i \in \calI'}a'_i
    =   \sum_{i \in \calI}(a_i + M) + (n - |\calI|) \cdot M
    =   \frac{S}{2} + nM.
    \]

    Therefore, the above constitutes a polynomial-time reduction from the partition problem to Balanced SSP, establishing that Balanced SSP is $\NP$-hard.
\end{proof}

\begin{proof}[Proof of \Cref{lem:integer-triple}]
    Fix integers $a, b, c \ge 0$. Let $\mu \coloneqq \frac{6a - 3b}{a + b + c}$, and define the \emph{cost} of the triple $(a, b, c)$ as
    \[
        \cost(a, b, c) \coloneqq (a + b) \cdot \max\{-\mu, 0\} + c \cdot \max\{2 - \mu, 0\}.
    \]
    
    Suppose that $(a, b, c)$ does not satisfy the last two conditions, i.e.,
    \[
        a \le \frac{b + c}{2}
    \quad \text{and} \quad
        \cost(a, b, c) < 1.
    \]
    It remains to show that $(a, b, c)$ must satisfy the first condition. To this end, we consider the following three cases based on the value of $\mu$.
    
    \paragraph{Case 1: $\mu \ge 2$.} By definition of $\mu$, we have
    \[
        \frac{6a - 3b}{a + b + c} \ge 2,
    \]
    which simplifies into $a \ge \frac{5}{4}b + \frac{1}{2}c$. Since we also have $a \le \frac{1}{2}b + \frac{1}{2}c$, it must be the case that $b = 0$ and $a = c / 2$. It follows that $(a, b, c)$ is proportional to $(1, 0, 2)$, and the first condition holds.
    
    \paragraph{Case 2: $0 \le \mu < 2$.} In this case, we have the following two inequalities:
    \begin{equation}\label{eq:triple-lemma-case-2}
        c \cdot \left(2 - \frac{6a - 3b}{a + b + c}\right) < 1
    \quad \text{and} \quad
        c \ge 2a - b \ge 0.
    \end{equation}
    The first follows from $\cost(a, b, c) < 1$. The inequality $c \ge 2a - b$ applies the assumption that $a \le \frac{b+c}{2}$, and the $2a - b \ge 0$ part follows from $\mu \ge 0$.
    
    Since $6a - 3b = 3 \cdot (2a - b) \ge 0$, the function $x \mapsto x \cdot \left(2 - \frac{6a - 3b}{a + b + x}\right)$ is increasing for $x \ge 0$. Then, we must have
    \[
        (2a - b) \left(2 - \frac{6a - 3b}{a + b + (2a - b)}\right) < 1,
    \]
    which simplifies into
    \[
        b \cdot (2a - b) < a.
    \]
    
    Since $b$ is an integer in $[0, 2a]$ and $1 \cdot (2a - 1) \ge a$ for any $a \ge 1$, we must have $b \in \{0, 2a\}$. If $b = 2a$, the first part of \Cref{eq:triple-lemma-case-2} reduces to $2c < 1$, which can only be satisfied by $c = 0$. This leads to a triple that is proportional to $(1, 2, 0)$. If $b = 0$, we obtain
    \[
        c \cdot \left(2 - \frac{6a}{a + c}\right) < 1
    \quad \text{and} \quad
        c \ge 2a.
    \]
    If $c = 2a$, we obtain a triple proportional to $(1, 0, 2)$. If $c \ge 2a + 1$, the first inequality above implies
    \[
        (2a + 1) \cdot \frac{2}{3a + 1} < 1,
    \]
    which cannot be satisfied by any integer $a \ge 0$.
    
    \paragraph{Case 3: $\mu < 0$.} In this case, the condition $\cost(a, b, c) < 1$ reduces to
    \[
        -(a+b)\mu + c\cdot(2 - \mu) < 1.
    \]
    Since $(a + b + c)\mu = 6a - 3b$, we have
    \[
        2c - (6a - 3b) < 1.
    \]
    Since the left-hand side above is an integer, it must be at most $0$. Then, we get
    \[
        6a - 3b \ge 2c \ge 0.
    \]
    On the other hand, for $\mu < 0$ to hold, we need $6a - 3b < 0$, a contradiction.
    
    Combining the three cases above shows that any triple that satisfies neither of the last two conditions must be proportional to either $(1, 2, 0)$ or $(1, 0, 2)$. This completes the proof.
\end{proof}

\section{Omitted Proof for \Cref{sec:one-sided-LB}}\label{app:one-sided-LB}

We derive \Cref{lem:anti-concentration} from the Berry-Esseen theorem.

\begin{theorem}[Berry-Esseen]\label{thm:berry-esseen}
    The following is true for some universal constant $C > 0$: For any i.i.d.\ zero-mean random variables $X_1, X_2, \ldots, X_n$ with $\Ex{}{X_i^2} = \sigma^2$ and $\Ex{}{|X_i|^3} = \rho$, it holds for any $a, b \in \bbR$ that
    \[
        \left|\pr{}{\frac{\sum_{i=1}^{n}X_i}{\sigma\sqrt{n}} \in [a, b]} - \pr{Z \sim \calN(0, 1)}{Z \in [a, b]}\right| \le \frac{C\rho}{\sigma^3\sqrt{n}}.
    \]
\end{theorem}

\begin{proof}[Proof of \Cref{lem:anti-concentration}]
    We apply the Berry-Esseen theorem with each $X_i$ being the distribution of $Z - p$ where $Z \sim \Bern(p)$. Note that $p \in [1/3, 2/3]$ implies
    \[
        \sigma^2 = p(1-p) \ge \frac{2}{9}
    \quad \text{and} \quad
        \rho = p \cdot (1-p)^3 + (1-p) \cdot p^3 \le \frac{1}{8}.
    \]
    It follows that the $C\rho / (\sigma^3\sqrt{n})$ term in \Cref{thm:berry-esseen} is at most $O(1/\sqrt{n})$.

    Since
    \[
        \pr{X \sim \Binomial(n, p)}{|X / n - p| \le \eps}
    =   \pr{}{\left|\frac{\sum_{i=1}^{n}X_i}{n}\right| \le \eps}
    =   \pr{}{\left|\frac{\sum_{i=1}^{n}X_i}{\sigma\sqrt{n}}\right| \le \frac{\eps\sqrt{n}}{\sigma}},
    \]
    it remains to show that
    \[
        \pr{Z \sim \calN(0, 1)}{|Z| \le \frac{\eps\sqrt{n}}{\sigma}} \le \frac{3}{\sqrt{\pi}} \cdot \eps\sqrt{n}.
    \]
    Since the density of $\calN(0, 1)$ is upper bounded by $1/\sqrt{2\pi}$, the left-hand side above is at most
    \[
        \frac{1}{\sqrt{2\pi}} \cdot 2 \cdot \frac{\eps\sqrt{n}}{\sigma}
    \le \frac{1}{\sqrt{2\pi}} \cdot 2 \cdot \frac{\eps\sqrt{n}}{\sqrt{2/9}}
    =   \frac{3}{\sqrt{\pi}} \cdot \eps\sqrt{n},
    \]
    as desired.
\end{proof}